\documentclass[a4paper,11pt]{article}
\pdfoutput=1
\usepackage{jheppub}
\usepackage{textcase}
\usepackage{graphicx,epsfig,psfrag,amssymb,hyperref}
\usepackage{multirow}
\usepackage{color,graphicx,epsfig,psfrag,amsmath,empheq}
\usepackage{booktabs}

\usepackage[dvipsnames]{xcolor}
\usepackage{bm}
\usepackage{mathrsfs,amsfonts,soul,color}
\usepackage[caption=false]{subfig}
\usepackage{hepunits}
\usepackage{enumitem}
\usepackage{placeins}
\usepackage{textcomp}
\usepackage{bbm}
\usepackage{verbatim}
\usepackage{amsthm}
\usepackage{thmtools} 
\usepackage{thm-restate}

\usepackage{float}
\usepackage[toc,page]{appendix}
\usepackage{algorithm}  
\usepackage{algorithmic}
\usepackage{wrapfig}

\newtheorem{theorem}{Theorem}
\newtheorem{lemma}{Lemma}

\newtheorem{example}{Example}

\newcommand{\argmin}[2]{\textrm{argmin}_{#1}~#2}

\newcommand{\E}[2]{\mathbb{E}_{#1}\left[#2\right]}

\newcommand\independent{\protect\mathpalette{\protect\independenT}{\perp}}
\def\independenT#1#2{\mathrel{\rlap{$#1#2$}\mkern2mu{#1#2}}}

\usepackage{xspace}

\newcommand{\X}[0]{\mathcal{X}}
\newcommand{\p}[0]{\mathcal{P}}

\newcommand{\Y}[0]{\mathcal{Y}}

\newcommand{\D}[0]{\mathcal{D}}
\newcommand{\N}[0]{\mathcal{N}}

\newcommand{\R}[0]{\mathbb{R}}

\newcommand{\ra}{\rightarrow}

\newif\ifsinglecolumn

\newcommand{\Ndata}[0]{n_{\mathrm{data}}}
\newcommand{\Pdata}[0]{\mathcal{P}_{\mathrm{data}}}

\newcommand{\dsim}[0]{\D^{\mathrm{sim}}}
\newcommand{\Nsim}[0]{n_{\mathrm{sim}}}
\newcommand{\Psim}[0]{\mathcal{P}_{\mathrm{sim}}}
\newcommand{\psim}[0]{p_{\mathrm{sim}}}

\newcommand{\cwola}{\textsc{CWoLa}\xspace}
\newcommand{\multicwola}{\textsc{Multi-CWoLa}\xspace}

\newcommand{\salad}{\textsc{SALAD}\xspace}
\newcommand{\multisalad}{\textsc{Multi-SALAD}\xspace}
\newcommand{\switchsalad}{\textsc{SALAD-Switch}\xspace}

\newcommand{\F}[0]{\mathcal{F}}
\newcommand{\G}[0]{\mathcal{G}}
\newcommand{\h}[0]{\mathcal{H}}

\begin{document}

\title{Resonant Anomaly Detection with Multiple Reference Datasets}

\author{Mayee F. Chen,$^{a}$}
\author{Benjamin Nachman,$^{b,c}$}
\author{and Frederic Sala$^{d}$}

\affiliation{
\phantom{ }\hspace{-0.12in}$^a$Computer Science Department, Stanford University, Stanford, CA 94305, USA \\
\phantom{ }\hspace{-0.12in}$^b$Physics Division, Lawrence Berkeley National Laboratory, Berkeley, CA 94720, USA \\
\phantom{ }\hspace{-0.12in}$^c$Berkeley Institute for Data Science, University of California, Berkeley, CA 94720, USA \\
\phantom{ }\hspace{-0.12in}$^d$Department of Computer Sciences, University of Wisconsin, Madison, WI 53706, USA \\
}

\emailAdd{mfchen@stanford.edu}
\emailAdd{bpnachman@lbl.gov}
\emailAdd{fredsala@cs.wisc.edu}

\abstract{
An important class of techniques for resonant anomaly detection in high energy physics builds models that can distinguish between reference and target datasets, where only the latter has appreciable signal. Such techniques, including Classification Without Labels (\cwola) and  Simulation Assisted Likelihood-free Anomaly Detection (\salad) rely on a single reference dataset. They cannot take advantage of commonly-available multiple datasets and thus cannot fully exploit available information. In this work, we propose generalizations of \cwola and \salad for settings where multiple reference datasets are available, building on weak supervision techniques. We demonstrate improved performance in a number of settings with realistic and synthetic data. As an added benefit, our generalizations enable us to provide finite-sample guarantees, improving on existing asymptotic analyses.
}

\maketitle

\clearpage
\section{Introduction}
\label{sec:intro}

Due to the vast parameter space of Standard Model extensions and to the lack of significant evidence for new particles or forces of nature, a new model-agnostic search paradigm has emerged.  Many of these \textit{anomaly detection} (AD) strategies are enabled by machine learning (see e.g. Ref.~\cite{hepmllivingreview,Karagiorgi:2021ngt,Kasieczka:2021xcg,Aarrestad:2021oeb}) and the first results with collision data are now available~\cite{collaboration2020dijet,ATLAS-CONF-2022-045}. One way to characterize AD methods is based on their physics assumption of the new phenomena~\cite{Karagiorgi:2021ngt}.  Strategies that assume the new physics is ``rare''~\cite{Kasieczka:2022naq} estimate (explicitly or implicitly) the data probability density and focus on events with low density.  In contrast, techniques that assume the new physics will manifest as an overdensity in phase space use likelihood ratio methods to compare a reference dataset to a target dataset. The latter approach has been extensively studied in the context of \textit{resonant anomaly detection}~\cite{2107.02821}, where one resonant feature (usually a mass) is used to create a sideband region (reference dataset) nearly devoid of any anomalous events and a signal region (target dataset) that may contain anomalies. The reference dataset is used to estimate the presence of anomalies in the target dataset via interpolation.

Many existing approaches are defined using one reference dataset and one target dataset~\cite{Collins:2018epr,Collins:2019jip,DAgnolo:2018cun,DAgnolo:2019vbw,1815227,Andreassen:2020nkr,Nachman:2020lpy,Amram:2020ykb,Hallin:2021wme,Raine:2022hht,dAgnolo:2021aun,Chakravarti:2021svb,Dillon:2022tmm,Letizia:2022xbe,2203.09601,Alvi:2022fkk,Raine:2022hht}. However, in practice one can have access to or construct \emph{multiple} references.  First, there may exist multiple resonant features that can be used to construct sideband and signal regions. For instance, when a particle decays into two new particles, the decay products can be used to construct all three intermediate resonances, a setting present in the LHC Olympics Dataset~\cite{Kasieczka:2021xcg}. Second, there may also exist multiple independent Standard Model simulators available for producing a dataset (e.g. Pythia~\cite{Sjostrand:2007gs}, Herwig~\cite{Bellm:2015jjp}, or Sherpa~\cite{Sherpa:2019gpd}). Using multiple reference datasets may improve performance, but it is not clear how to incorporate all of their information when using existing methods designed for a single set.

We explore two generalizations of resonant AD to multiple reference datasets. First, we consider Classification Without Labels (\cwola)~\cite{Metodiev:2017vrx,Collins:2018epr,Collins:2019jip}, in which the reference is simply the sideband region---a form of weak supervision where the noisy label of ``signal'' is assigned to events in the signal region and the noisy label of `background' to events in the sideband region. 
We propose a new method, \multicwola, that builds multiple reference datasets by constructing signal and sideband regions along different resonant features. We consider a point's membership in each feature's signal region as a noisy vote for anomaly, learn weights on each vote, and aggregate them to produce a higher-quality noisy label. We demonstrate \multicwola's performance on the LHC Olympics Dataset~\cite{Kasieczka:2021xcg}.

Next, we study Simulation Assisted Likelihood-free Anomaly Detection (\salad)~\cite{Andreassen:2020nkr}. In this method, a reweighting function between a reference simulation dataset and a target dataset is learned in the sideband conditioned on the resonant feature. The simulated events in the signal region are reweighted by interpolating this function and then are used to distinguish anomalies in the target dataset.  
We extend this to the case of multiple simulated datasets, each of which may make different approximation choices and thus provide complementary accuracy when using \salad. We introduce \multisalad, which combines the simulated datasets accordingly and then reweights, with the key finding that combining data helps when each simulator approximates different components of the background well. We demonstrate \multisalad's performance on synthetic data.

Finally, we study the finite sample guarantees of our proposed methods. Many resonant AD methods have optimality guarantees in some asymptotic limit, but there is no first-principles understanding of the methods' performance with finite samples.  
In particular, approaches like the ones described above that use classifiers to distinguish a reference dataset from a target dataset approximate the data-to-background likelihood ratio.  When the reference (physics) model is correct, this approach will converge to the optimal Neyman-Pearson likelihood ratio test in the limit of infinite statistics, complex enough classifier architecture, and flexible enough training procedure~\cite{neyman1933ix,Nachman:2020lpy}.  However, a finite sample understanding of these approaches is lacking.
We draw on results from statistical theory to begin a formal study of resonant AD methods with limited data. Our results lay a foundation for future investigations into the finite sample properties of AD and related methods.

This paper is organized as follows.  Section~\ref{sec:setup} briefly set up the resonant AD setting and then \multicwola and \multisalad are introduced in Secs.~\ref{sec:method} and~\ref{sec:salad}, respectively.  The paper ends with conclusions and outlook in Sec.~\ref{sec:conclusion}.

\section{Problem Setup}
\label{sec:setup}

We have an input space of \emph{discriminating features} $x \in \X$ and $k$ \emph{resonant features} $m = [m^1, \dots, m^k] \in \R^k$. Associated with a point $(x, m)$ is an unknown label $y \in \Y$ for $\Y = \{0, 1\}$ (background vs. signal). Points $(x, m, y)$ are drawn from a distribution $\mathcal{P}$ with density $p(\cdot)$. 
For a resonant feature $m^i \in \R$, an interval $\mathcal{I}_{m^i} \subset \R$ is used to define a signal region $SR_i = \{(x,m): m^i \in \mathcal{I}_{m^i}\}$ and a sideband region $SB_i = \{(x,m): m^i \notin \mathcal{I}_{m^i}\}$ (when the resonant feature is obvious, the $i$ is dropped and we use $SR$ and $SB$). We assume that the sideband region contains little to no signal, i.e., $p(y = 1 | (x, m) \in SB) \approx 0$. Our goal is to construct a predictor $f: \X \ra \Y$ to perform anomaly detection.

\section{\multicwola: Learning from Multiple Resonant Features}
\label{sec:method}
We introduce \multicwola, an approach to anomaly detection that uses multiple reference datasets and is built using principles from the area of weak supervision \cite{Ratner18, Fu20}.

\paragraph{Standard \cwola} We have one unlabeled dataset $\D = \{(x_i, m_i)\}_{i = 1}^n$ with one resonant feature ($k = 1$) that we want to use to learn $f$. We use $m$ to construct the signal and sideband regions, $\D_{SR}, \D_{SB} \subset \D$ where $\D_{SR} = \D \cap SR$ and $\D_{SB} = \D \cap SB$, with distributions $p_{SR}$ and $p_{SB}$ respectively. With the intuition that there are more anomalies in the signal region, we express each distribution as a mixture of signal and background components with weight $0\leq \eta_{SR}, \eta_{SB} \leq 1$:
\begin{align}
    p_{SR}(x) = \eta_{SR} p(x | y = 1) + (1 - \eta_{SR}) p(x | y = 0) \\
    p_{SB}(x) = \eta_{SB} p(x | y = 1) + (1 - \eta_{SB}) p(x | y = 0)
\end{align}

Under this construction, the density ratio of the mixtures $\frac{p_{SR}(x)}{p_{SB}(x)}$ can be written in terms of the ratio of the signal and background components, $r(x) = \frac{p(x | y = 1)}{p(x | y = 0)}$, as $\frac{p_{SR}(x)}{p_{SB}(x)} = \frac{\eta_{SR} r(x) + 1 - \eta_{SR}}{\eta_{SB} r(x) + 1 - \eta_{SB}}$. Assuming $\eta_{SR} > \eta_{SB}$ (e.g. more signal in the signal region), the mixture ratio is monotonically increasing in $r(x)$. Therefore, we train a classifier $f$ to learn $\frac{p_{SR}(x)}{p_{SB}(x)}$ by distinguishing between $\D_{SR}$ and $\D_{SB}$, and this $f$ provides information about $r(x)$ and can be used for anomaly detection.

\subsection{Multi-CWoLa Method}

Intuitively, \cwola uses the resonant feature $m$ as a noisy label that identifies the signal versus sideband region and then trains a classifier using these. This idea leads to a simple question---if more than one such feature is available $(k > 1)$, how can the multiple noisy labels best be utilized? We tackle this question using principles from weak supervision~\cite{Ratner18, ratner2016data, ratner2019training, Fu20}.

\subsubsection{Model}

In our approach, we split $\D$ along each resonant feature $m^i$ to produce  pairs of datasets $\D_{SB_i}$ and $\D_{SR_i}$ for each $i \in [k]$ based on membership in $I_{m^i}$.
A straightforward way to use all datasets $(\D_{SB_1}, \D_{SR_1}), \dots, (\D_{SB_k}, \D_{SR_k})$ is to apply standard \cwola $k$ times by training $k$ classifiers that we can then ensemble or average.
Instead, in \multicwola, we construct a binary vector per $x$ consisting of $k$ noisy membership labels, $\mathbf{M}(m) = \{M_1(m), \dots, M_k(m)\} \in \{0, 1\}^k$, where $M_i(m) = 1$ if $(x, m) \in \D_{SR_i}$ and $M_i(m) =0$ if $(x, m) \in \D_{SB_i}$. We propose to directly aggregate these labels $\mathbf{M}(m)$ into an estimate of $y$, $\hat{y}$, and train a classifier on the aggregated $\hat{y}$ along with the discriminative features $x$. Since each $M_i(m)$'s ``vote'' can have different correlation with the true $y$, we aim to combine the votes in a weighted fashion. We cannot directly measure each membership label's accuracy since the true $y$ is unknown, so we draw on methods from weak supervision.

We model the distribution $p(y, \mathbf{M}(m))$ as a probabilistic graphical model with the following parametrization:
\begin{align}
    p(y, \mathbf{M}(m); \theta) = \frac{1}{Z} \exp \bigg(\theta_y \widetilde{y} + \sum_{i = 1}^k \theta_i \widetilde{M}_i(m) \widetilde{y} \bigg), \label{eq:pgm}
\end{align}

where $\theta = \{\theta_y, \theta_i \; \forall i \in [k] \}$ are the canonical parameters of the distribution, $Z$ is for normalization, and $\widetilde{y}$ and $\widetilde{M}_i(m)$ are $y$ and $M_i(m)$ scaled from $\{0, 1\}$ to $\{-1, 1\}$. Intuitively, $\theta_i$ represents the (unobserved) strength of the correlation between $M_i(m)$ and $y$ and thus captures a notion of $M_i$'s accuracy. This model also implies, for simplicity, that $M_i(m) \perp M_j(m) | y$; that is, the resonant features are conditionally independent given $y$.\footnote{We can model some dependencies among resonant features if desired (see~\cite{Fu20} for a method and see~\cite{varma2019learning} for how to learn if resonant features are not conditionally independent). However, we need at least three conditionally independent subsets of resonant features in $\mathbf{M}(m)$ in order for the estimation method from~\cite{Fu20} to recover the correct parameters.}

Our goal is to estimate the parameters of the graphical model and use them to perform inference, producing aggregated weak labels $\hat{y}$ from the distribution $p(y = 1 | \mathbf{M}(m); \theta)$ given a vector of noisy labels $\mathbf{M}(m)$.

\subsubsection{Parameter Estimation}

We first learn the parameters of $p(y, \mathbf{M}(m); \theta)$ as defined in~\eqref{eq:pgm}.
Of key interest is the accuracy parameter $\alpha_i = p(M_i(m) = 1 | y = 1) = p(M_i(m) = 0 | y = 0)$ of the $i$th resonant feature, which corresponds to the canonical parameter $\theta_i$ (see~\cite{wainwright2008graphical} for more background on probabilistic graphical models). 
We estimate the accuracy parameters by adapting the \emph{triplet} approach from \cite{Fu20}.  First, we draw triplets of resonant features $a, b, c \in [k]$.\footnote{We assume that $k \ge 3$. In Lemma~\ref{lemma:small_k}, we discuss why having $k = 1$ or $k = 2$ resonant features does not recover a unique model.} If the distribution on $y, \mathbf{M}(m)$ follows the graphical model in~\eqref{eq:pgm}, it holds that $y M_a(m) \perp y M_b(m) $ if $M_a(m) \perp M_b(m) | y$. Then, we have that $\mathbb{E}[\widetilde{y}\widetilde{M}_a(m)]\mathbb{E}[\widetilde{y}\widetilde{M}_b(m)] = \mathbb{E}[\widetilde{M}_a(m)\widetilde{M}_b(m)]$ since $\widetilde{y}^2 = 1$. Writing one such equation for each pair in the triplet $(a, b, c)$, we have that 
\begin{align*}
    \mathbb{E}[\widetilde{y}\widetilde{M}_a(m)]\mathbb{E}[\widetilde{y}\widetilde{M}_b(m)] &= \mathbb{E}[\widetilde{M}_a(m)\widetilde{M}_b(m)]\\
    \mathbb{E}[\widetilde{y}\widetilde{M}_a(m)]\mathbb{E}[\widetilde{y}\widetilde{M}_c(m)] &= \mathbb{E}[\widetilde{M}_a(m)\widetilde{M}_c(m)]\\
    \mathbb{E}[\widetilde{y}\widetilde{M}_b(m)]\mathbb{E}[\widetilde{y}\widetilde{M}_c(m)] &= \mathbb{E}[\widetilde{M}_b(m)\widetilde{M}_c(m)].
\end{align*}
Solving this system, we obtain
\begin{align*}
    |\mathbb{E}[\widetilde{y}\widetilde{M}_a(m)]| = \sqrt{\bigg|\frac{\mathbb{E}[\widetilde{M}_a(m)\widetilde{M}_b(m)]\mathbb{E}[\widetilde{M}_a(m)\widetilde{M}_c(m)]}{\mathbb{E}[\widetilde{M}_b(m)\widetilde{M}_c(m)]}\bigg|},
\end{align*}
and similarly for $b$ and $c$.
We assume that each signal region is positively correlated with the true signal, which allows for us to ignore the absolute value and uniquely recover $\mathbb{E}[\widetilde{y}\widetilde{M}_a(m)]$. Next, we can use $\mathbb{E}[\widetilde{y}\widetilde{M}_a(m)] = 2 p(\widetilde{y} = \widetilde{M}_a(m)) - 1$ to obtain $\alpha_i$ using properties of the graphical model in~\eqref{eq:pgm}. Note that in practice, all of these quantities are empirical estimates, with terms such as $\hat{\mathbb{E}}[\widetilde{M}_a(m)\widetilde{M}_b(m)] = \frac{1}{n} \sum_{i=1}^n \widetilde{M}_a(m_i)\widetilde{M}_b(m_i)$.

\subsubsection{Inference and Training}
After we learn the accuracy parameters, we use them to estimate $p(y = 1| \mathbf{M}(m))$ for a given $\mathbf{M}(m)$. We use Bayes' rule and the conditional independence among $\mathbf{M}(m)$ to write $p(y | \mathbf{M}(m)) = \frac{\prod_{i = 1}^mp(M_i(m) | y = 1) p(y = 1)}{p(\mathbf{M}(m))}$. We assume that the class balance $p(y = 1)$ is known; otherwise, it can be estimated via tensor decomposition~\cite{ratner2019training}. $p(M_i(m) | y = 1)$ is either equal to $\alpha_i$ if $M_i(m) = 1$ or $1 - \alpha_i$ if $M_i(m) = 0$, and the denominator $p(\mathbf{M}(m))$ can be either directly estimated since all quantities are observable or computed as $\prod_{i = 1}^m p(M_i(m) | y = 1) p(y = 1) + \prod_{i = 1}^m p(M_i(m) | y = 0) p(y = 1)$ using the estimated accuracies and class balance. 

Once $p(y = 1 | \mathbf{M}(m))$ is estimated for all $\mathbf{M}(m) \in \{0, 1\}^k$, the aggregated weak label $\hat{y}$ is drawn from such distribution. With labels $\hat{y}$ for each $(x, m) \in \D$, we train a classifier $\hat{f}$ on the weakly labeled dataset $\{(x, \hat{y})\}_{i = 1}^n$. This procedure is summarized in Algorithm~\ref{alg:multicwola}.

\begin{algorithm}[t]
\caption{\multicwola}
\begin{algorithmic}[1]
\STATE \textbf{Input:} Dataset $\D = \{(x_i, m_i)\}_{i = 1}^n$; thresholds $\mathcal{I}_{m^i}$ that split $\D$ into signal and sideband regions, $\D_{SR_i}$ and $\D_{SB_i}$ respectively, for each $m^i$; class balance probability of anomaly $p(y = 1)$
\STATE Construct noisy label $M_i(m) = \begin{cases} 1 & (x, m) \in \D_{SR_i} \\ 0 & (x, m) \in \D_{SB_i} \end{cases}$ for each resonant feature $m^i$.
\FOR{each triplet $a, b, c \in [k]$}
\STATE   \begin{align}
        \beta_a &:= \sqrt{\big|\hat{\mathbb{E}}[\widetilde{M}_a(m) \widetilde{M}_b(m)] \hat{\mathbb{E}}[\widetilde{M}_a(m) \widetilde{M}_c(m)] / \hat{\mathbb{E}}[\widetilde{M}_b(m) \widetilde{M}_c(m)] \big|} \\
        \beta_b &:= \sqrt{\big|\hat{\mathbb{E}}[\widetilde{M}_a(m) \widetilde{M}_b(m)] \hat{\mathbb{E}}[\widetilde{M}_b(m) \widetilde{M}_c(m)] / \hat{\mathbb{E}}[\widetilde{M}_a(m) \widetilde{M}_c(m)] \big|} \\
        \beta_c &:= \sqrt{\big|\hat{\mathbb{E}}[\widetilde{M}_a(m) \widetilde{M}_c(m)] \hat{\mathbb{E}}[\widetilde{M}_b(m) \widetilde{M}_c(m)] / \hat{\mathbb{E}}[\widetilde{M}_a(m) \widetilde{M}_b(m)] \big|},
    \end{align}
 where $\hat{\mathbb{E}}$ is an empirical estimate of the expectation over $\D$, and $\widetilde{M}(m)$ indicates $M(m)$ scaled to $\{-1, 1\}$.
\ENDFOR
\STATE Set accuracy parameter $\alpha_i = \hat{p}(M_i(m) = 1 | y = 1) = \hat{p}(M_i(m) = 0 | y = 0) = \hat{p}(M_i(m) = y) = \frac{\beta_i + 1}{2}$.
\STATE Compute estimate $\hat{p}(y = 1 | \mathbf{M}(m)) \propto \prod_{i = 1}^m \hat{p}(M_i(m) | y = 1) p(y = 1)$.
\STATE Construct $\hat{y} \sim \hat{p}(y = 1 | \mathbf{M}(m))$ for each $(x, m) \in \D$.
\STATE \textbf{Output:} Classifier $\hat{f}$ for anomaly detection trained on $\{(x_i, m_i, \hat{y}_i)\}_{i = 1}^n$.
\end{algorithmic}
\label{alg:multicwola}
\end{algorithm}

\subsection{Theoretical Results}

Under~\eqref{eq:pgm}, \multicwola offers \emph{finite-sample generalization} guarantees.
Suppose the downstream model $\hat{f}$ trained on $\hat{y}$ belongs to class $\mathcal{F}$. Define a loss function $\ell_C:\Y \times \Y \ra \R$ and let the expected loss of $f$ be $L_C(f) := \E{}{\ell_C(f(x), y)}$ on true labels. Then, the optimal classifier is $f^\star = \argmin{f \in \F}{L_C(f)}$, which is achieved with unlimited labeled data. Let the empirical loss of $f$ on $\hat{y}$ be $\hat{L}_C(f) := \frac{1}{n} \sum_{i = 1}^n \ell_C(f(x_i), \hat{y}_i)$. Then, the $\hat{f}$ we learn is constructed from $\hat{f} = \argmin{f \in \F}{\hat{L}_C(f)}$, which is learned on finite and noisily labeled data. Note that this construction is different from the standard empirical risk minimization (ERM) loss on labeled data, and thus $\hat{L}_C(f)$ does not asymptotically equal $L_C(f)$. We aim to minimize the \textit{generalization error} $L_C(\hat{f}) - L_C(f^\star)$.

We now present our result on an upper bound for $L_C(\hat{f}) - L_C(f^\star)$.
Define the Rademacher complexity of $\F$ as $\mathfrak{R}_n(\ell \circ \F) = \E{}{\sup_{f \in \mathcal{F}} \frac{1}{n} \sum_{i=1}^n \varepsilon_i \ell(f(x_i), y_i)}$ with random variables $\Pr(\varepsilon = 1) = \Pr(\varepsilon = -1) = \frac{1}{2}$. Define $e_{\min}$ as the minimum eigenvalue of the covariance matrix on $[y, M_1(m), \dots, M_k(m)]$, and let $a_{\min}$ be the minimum value of $\mathbb{E}[\widetilde{M}_i(m) y]$ over all $i$.
\begin{theorem}
Assume that $p(y, \mathbf{M}(m))$ can be parametrized according to~\eqref{eq:pgm} and that $\ell$ is scaled to be bounded in $[0, 1]$. Assume that the class balance $p(y)$ is known (if not, there are ways to estimate it~\cite{ratner2019training}), and that $k \ge 3$. Then, with probability at least $1 - \delta$, the generalization error of \multicwola on $\D$ is at most
\begin{align*}
    L_C(\hat{f}) - L_C(f^\star) &\le 4 \mathfrak{R}_n(\ell \circ \F) + \sqrt{\frac{\log 2/\delta}{2n}} + \frac{c_1}{e_{\min} a_{\min}^5} \Big(\sqrt{\frac{k}{n}} + \frac{c_2 k}{\sqrt{n}} \Big), 
\end{align*}

where $c_1, c_2$ are positive constants. 
\label{thm:multicwola}
\end{theorem}

\noindent We observe that there are three quantities controlling the above bound:
\begin{itemize}[leftmargin=0.5cm]
    \item The \textit{Rademacher complexity} of $\F$: this term describes the model's expressivity. Smaller Rademacher complexity means that the model is easier to learn and that our $\hat{f}$ will be closer to the best model in $\F$. This quantity can be readily computed for a variety of function classes $\F$, such as decision trees, linear models, and two-layer feedforward networks, which makes our bound in Theorem~\ref{thm:multicwola} tractable. See Appendix~\ref{supp:rademacher} for exact values.
    \item Using $n$ finite samples: as the amount of data increases, the error decreases in $\mathcal{O}(n^{-1/2})$.
    \item Using noisy labels $\hat{y}$ instead of $y$: for our weak supervision algorithm and graphical model, using $\hat{y}$ rather than $y$ contributes an additional $\mathcal{O}(n^{-1/2})$ error. Asymptotically, our approach thus does no worse than training with labeled data.
\end{itemize}

By contrast, the standard \cwola approach with $k = 1$ does not utilize any aggregation or weak supervision, which requires $k \ge 3$. For standard \cwola, the second term in the generalization error is irreducible due to the fact that using any single resonant feature in place of $y$ is biased. On the other hand, \multicwola corrects for some of this bias; the second term asymptotically approaches $0$ with more data. 

\subsection{Empirical Results}
In Figure~\ref{fig:snorkelcwola}, we compare \multicwola with standard \cwola as well as two other baselines. We use simulation data from the LHC Olympics Dataset~\citep{Kasieczka:2021xcg}; in particular from Pythia 8~\citep{Sjostrand:2007gs}, where the signal is boson decay and the background is generic $2 \rightarrow 2$ parton scattering. This dataset contains $5$ features; in the standard \cwola setup, we use one thresholded resonant feature  ($k = 1$) and use $4$ discriminative features as $x$. For \multicwola, we have generated $k = 3$ mixtures by varying how the $3$ resonant features (the jet masses in addition to the dijet mass) are thresholded and use $2$ discriminative features as $x$.  
We have three other baselines that utilize $3$ resonant features: 
\begin{itemize}
    \item \cwola + intersect defines the signal region as the intersection of the resonant features' signal regions, e.g. $SR = SR_1 \cap SR_2 \cap SR_3$, but this can be overly conservative.
    \item \cwola $+ x$ thresholding has one resonant feature as the noisy label $\hat{y} = M_1(m)$, and includes the remaining thresholded features as discriminative $\{M_2(m), M_3(m), x\}$.
    \item \cwola + average runs standard \cwola three times, once per resonant feature and with the $2$ discriminative features. The three model scores are averaged to produce the final output.
\end{itemize}

We vary the number of samples available on a logarithmic scale from $n = 59$ to $6003$ and plot the AUC averaged over $5$ runs per sample size in~\ref{fig:snorkelcwola}. We find that \multicwola offers a higher AUC and lower variance, especially when there is limited data. We also plot the SI curves averaged over $5$ runs for $n = 59, 530, 6003$ in~\ref{fig:cwolasic}.

\begin{figure}
    \centering
    \includegraphics[width=0.8\textwidth]{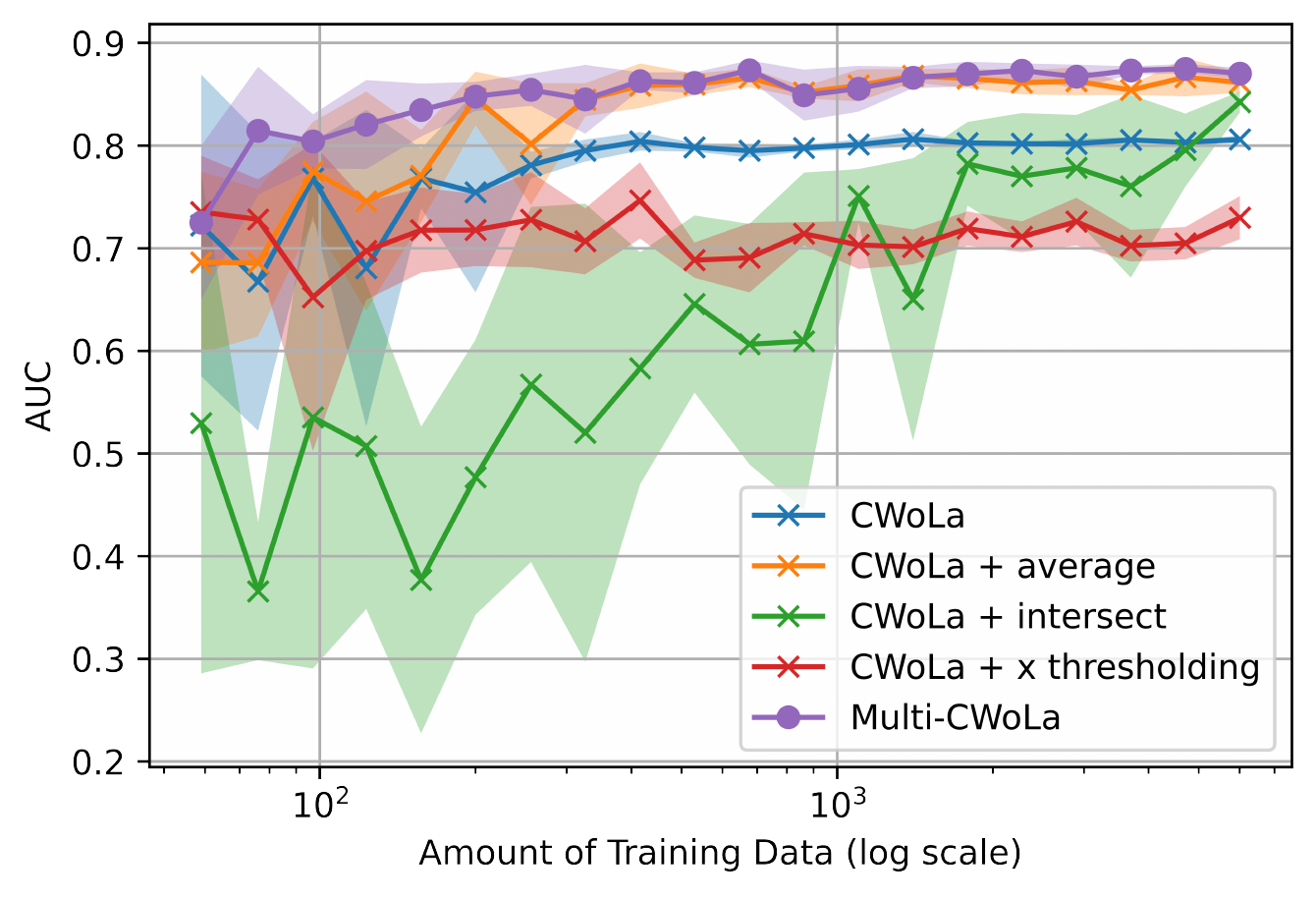}
    \caption{Comparison between \cwola and \multicwola. Using multiple mixed samples helps performance across a range of dataset sizes. Access to multiple weak sources enables better AUC and lower variance compared to the single-feature version.}
    \label{fig:snorkelcwola}
\end{figure}

\begin{figure}
    \centering
    \includegraphics[width=0.45\textwidth]{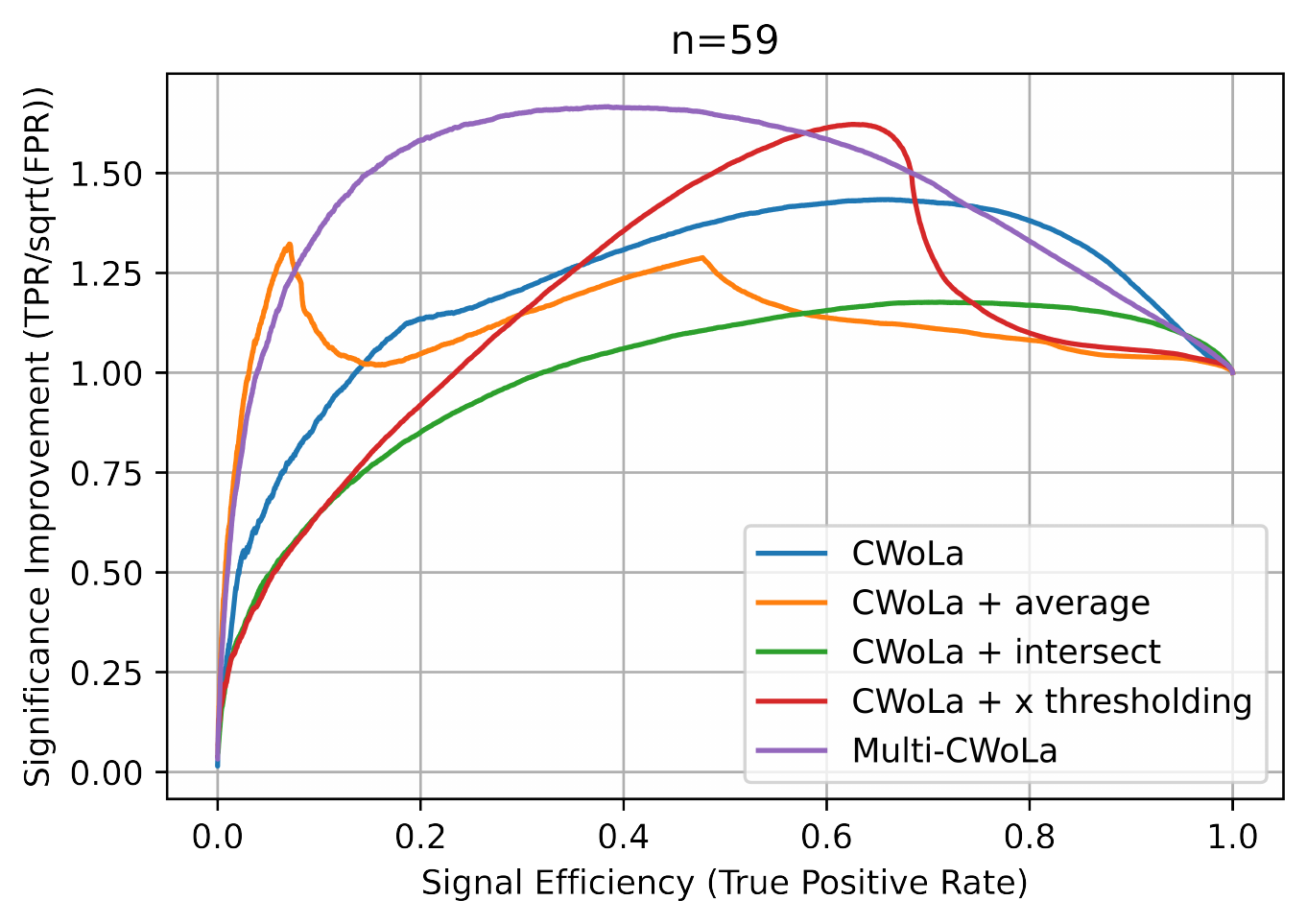}
    \includegraphics[width=0.45\textwidth]{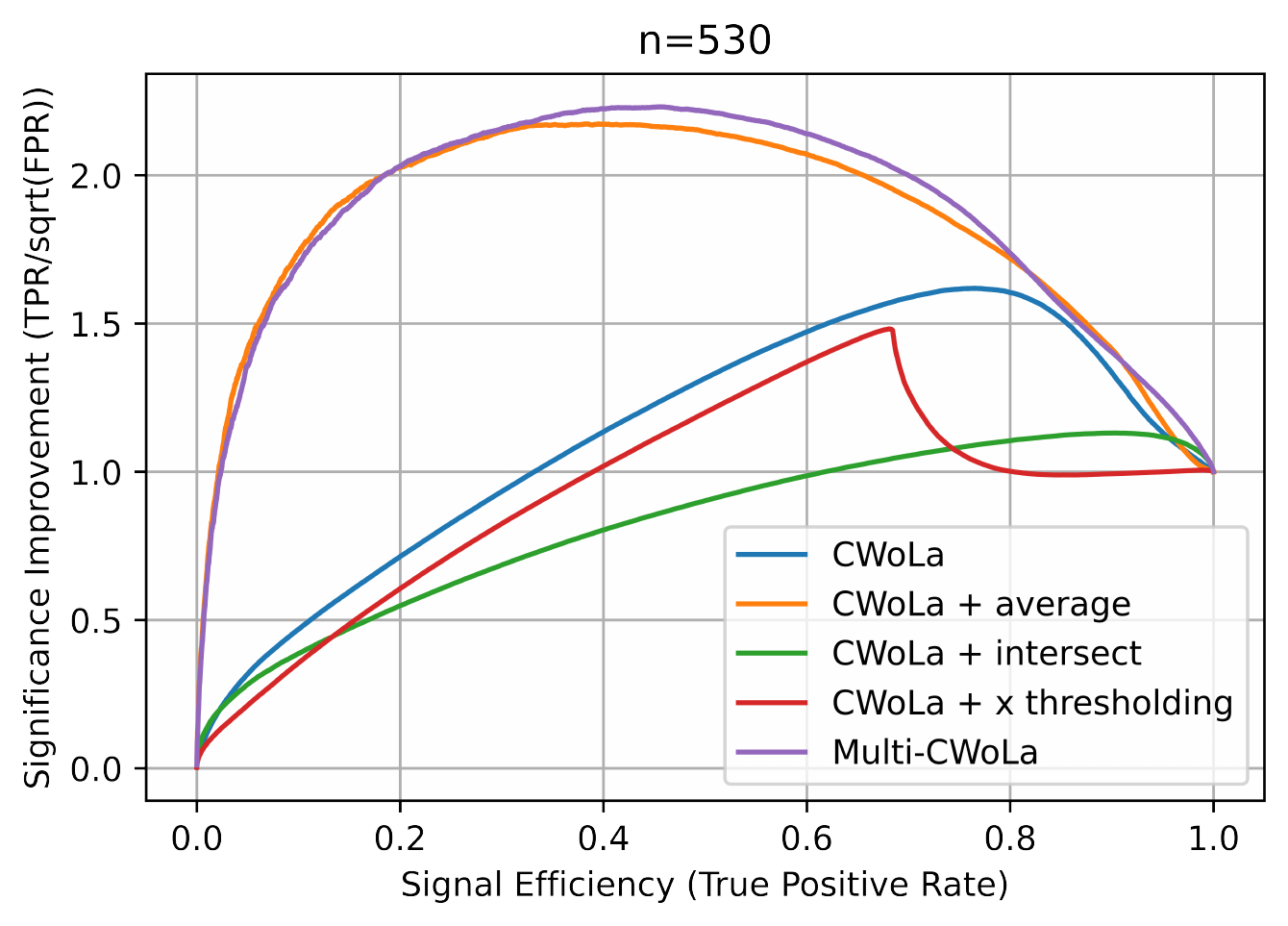}
    \includegraphics[width=0.45\textwidth]{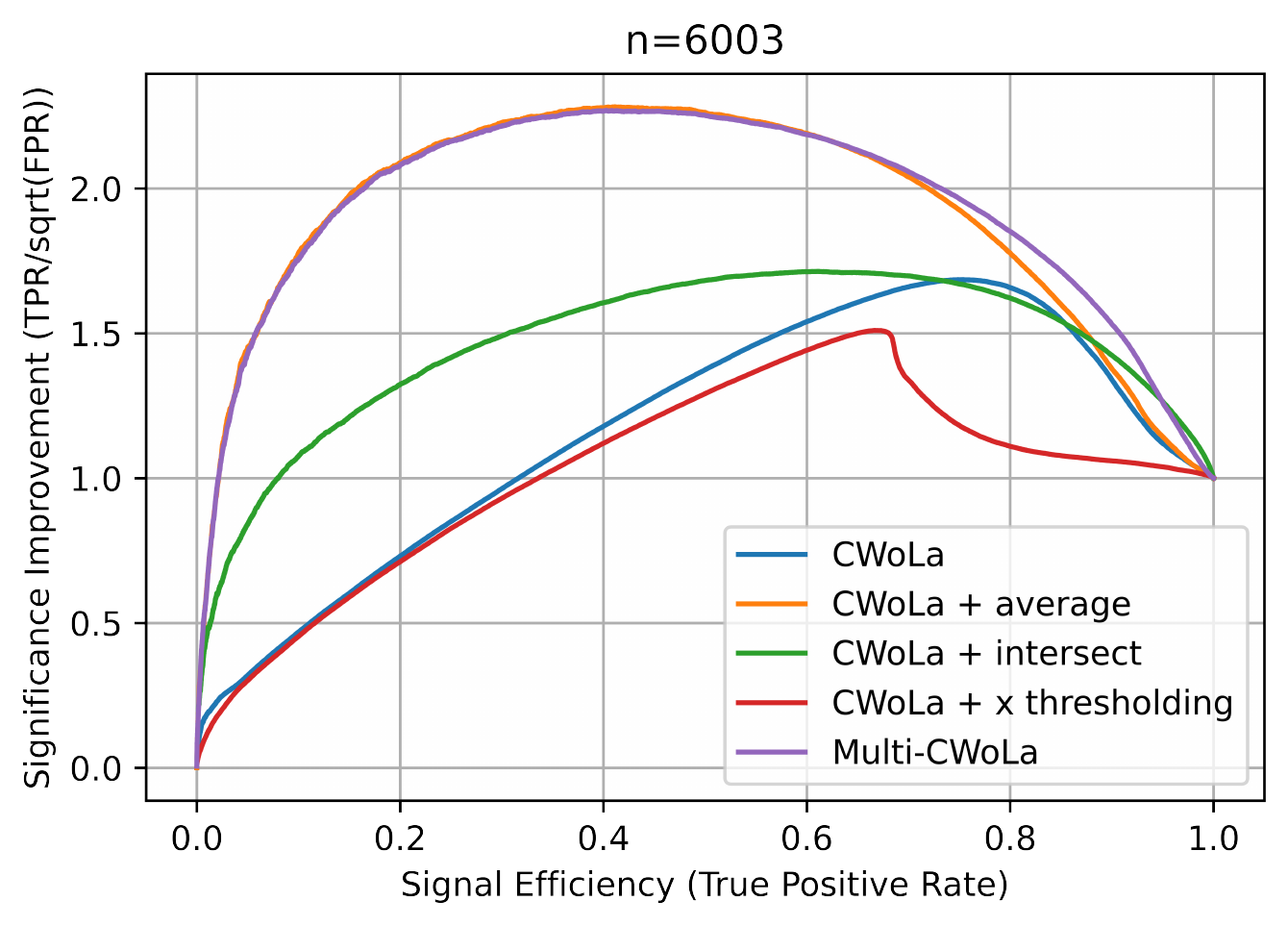}
    \caption{Significance Improvement (SI) curve for \multicwola at sizes $n = 59, 530,$ and $6003$.}
    \label{fig:cwolasic}
\end{figure}

\section{\multisalad: Learning from Multiple Simulations} \label{sec:salad}

We often have access to a(n approximate) simulation of the background process. We first provide an overview of \salad, which reweighs samples from the simulation to better assist with classification on the real dataset. Then, we present \multisalad, a variant of SALAD that uses multiple simulations.

\textbf{Standard \salad} We have a background simulation dataset $\dsim = \{(x_i, m_i)\}_{i = 1}^{\Nsim}$ with $y_i = 0$ for all $i$ in addition to one true dataset $\D = \{(x_i, m_i)\}_{i = 1}^{n}$. $\dsim$ is drawn from some distribution $\Psim$ with density $\psim$.
While CWoLA learns the likelihood ratio between the signal and sideband regions of $\D$ alone, \salad utilizes $\dsim$ as well. Note that if $\psim$ is equal to $p(\cdot | y = 0)$, we could directly train a model to distinguish between $\D$ and $\dsim$ in the signal region to get a classifier that could detect anomalies. However, since $\dsim$ may not match the true background data, we instead first need to learn a reweighting function that captures the differences between $\dsim$ and $\D$'s background data, and then we train a model to distinguish between $\D$ and the reweighted $\dsim$ in the signal region.
Formally, given fixed $SR$ and $SB$ for both datasets, the method can be broken into two steps:
\begin{enumerate}
    \item \textbf{Reweighting}: a classifier $\hat{g}$ is trained to distinguish between $\dsim_{SB} = \dsim \cap SB$ and $\D_{SB}$. Assuming that the sideband region has no anomalies, this $\hat{g}$ is able to produce an estimate of the weight ratio\footnote{This is with the binary cross entropy loss function (also works for other functions~\cite{Nachman:2021yvi}).  This likelihood-ratio trick is well-known (see e.g. Ref.~\cite{hastie01statisticallearning,sugiyama_suzuki_kanamori_2012}), also in high-energy physics (see e.g. Ref.~\cite{Cranmer:2015bka}).} $w(x, m) = \frac{p(x, m | y = 0)}{\psim(x, m | y = 0)} \approx \frac{\hat{g}(x, m)}{1 - \hat{g}(x, m)}$, assuming that the datasets are the same size ($|\dsim_{SB}| = |\D_{SB}|$).
    \item \textbf{Detection}: Using a loss function $L_S$ with estimated $\hat{w}(x, m)$ applied to $\dsim_{SR} = \dsim \cap SR$, a classifier $\hat{h}$ is trained to distinguish between $\D_{SR}$ and $\dsim_{SR}$.
\end{enumerate}

If the estimate $\hat{w}(x, m)$ is exactly equal to $w(x, m)$ (e.g. $\hat{g}$ is Bayes-optimal), then the second step will be equivalent in expectation to learning the ratio $\frac{p(x)}{p(x | y = 0)}$ (see Lemma~\ref{lemma:salad_loss} in Appendix~\ref{supp:salad_loss}), from which one can detect anomalies.

\subsection{\multisalad Method}
Now, we have multiple simulation datasets $\dsim_1, \dots, \dsim_k$. 
One approach would be to maintain distinctions among simulations by reweighing each pair to learn $k$ weight functions $w_i(x, m)$, and then using one overall loss function that weights points from each $\dsim_{SR, i}$ with $w_i$.
However, it has been shown that importance reweighting, despite working in expectation, can be highly unstable and result in poor performance of tasks on the target data $\D$~\citep{dasgupta2003boosting}. To understand why, Ref.~\cite{cortes2010learning} showed that the generalization error of an empirical loss function with importance weights $w$ depends on the magnitude of $w$.
Applied to our setting, it suggests that the more inaccurate the simulation is, the less the reweighted loss recovers the true $\frac{p(x)}{p(x | y = 0)}$,
and the model may instead pick up on differences between $\D_{SR}$ and the reweighted $\dsim_{SR}$ that are noise rather than the anomaly. As a result, aggregating individual \salad outputs can be equivalent to ensembling many poor classifiers. 

Given these observations, \multisalad uses multiple simulation datasets in a very simple yet theoretically principled way: control the magnitude of the overall $w$ by combining all the $\dsim_i$ to produce one large simulation dataset $\widetilde{\D}^{\text{sim}}$ whose distribution best approximates the true background $p(x | y = 0)$, and then use standard \salad with $\widetilde{\D}^{\text{sim}}$ and $\D$. Note that this approach both improves sample complexity and can ``suppress'' a simulation that on its own has high $w$, while the approach of learning $k$ weight functions would not offer such improvements.
In Algorithm~\ref{alg:multisalad} and Appendix~\ref{supp:multisalad}, we write this procedure out where we simply concatenate all $\dsim_i$ together. However, with domain knowledge on the strengths and weaknesses of each simulation across features, one could produce $\widetilde{D}^{\text{sim}}$ by sampling accordingly from each. We leave this direction for future work.

\subsection{Theoretical Results}

We now present a finite sample generalization error bound on \multisalad that also applies to \salad. To measure the generalization error, recall $w(x, m) = \frac{p(x, m | y = 0)}{\psim (x, m | y = 0)}$ and let $\hat{w}$ be the classifier $g$'s estimate.
We denote $h$ as the reweighted classifier. Let $h^\star = \argmin{h \in \h}{L_S(h, w)}$ and let $\hat{h} = \argmin{h \in \h}{\hat{L}_S(h, \hat{w})}$. We aim to bound $L_S(\hat{h}, \hat{w}) - L_S(h^\star, w)$.

We first set up some definitions. Define $n^{SR}$ as the number of points from $\D$ and $\widetilde{\D}^{\text{sim}}$ belonging to the signal region, and $n^{SB}$ as the number of points belonging to the sideband. Let $\Nsim^{SR}$ be the number of points in $\widetilde{\D}^{\text{sim}}$ belonging to the signal region. Let $\hat{g}(x) \in [\hat{g}_{\min}, \hat{g}_{\max}]$ and $g^\star(x) \in [g^\star_{\min}, g^\star_{\max}]$, where $g^\star$ is the optimal classifier. Let $\mathfrak{R}_{n^{SR}}(\ell_S \circ \{H, G\})$ be the Rademacher complexity of the overall loss $L_S(h, w)$ across function classes $h \in \h, g \in \G$. Define $W = \max_{x, m} w(x, m)$ as the maximum ratio between the simulation and true background. Let $B_1 = \max \{-\log h^\star(x, m), -\log(1 - h^\star(x, m))\}$ be based on the most extreme value of $h^\star$ (i.e. how far apart $p$ and $p(\cdot | y = 0)$ can be). Let $\eta = \max (-\log (1 - h^\star(x, m)))$ for $x, m \in \dsim_{SR}$. Let $\mathfrak{R}_{n^{SB}}(\ell \circ \G)$ is the Rademacher complexity of the loss function class used for learning the reweighting, where $\ell$ is point-wise cross-entropy. Finally, let $B_2 = -\log (\min \{\hat{g}_{\min}, g^\star_{\min}\})$.

\begin{restatable}[]{theorem}{multisaladthm}
With probability at least $1 - \delta$, there exists a constant $c > 0$ such that the generalization error of \multisalad on $\widetilde{D}^{\text{sim}}$ and $\D$ is at most
\begin{align}
    L_S(\hat{h}, \hat{w}) &- L_S(h^\star, w) \le 2\mathfrak{R}_{n^{SR}}(\ell_S \circ \{ \mathcal{H}, \mathcal{G}\}) + (1 + WB_1) \sqrt{\frac{\log 8 / \delta}{2n^{SR}}} \\
    &+ \frac{\eta \Nsim^{SR}}{(1 - \hat{g}_{\max})(1 - g^\star_{\max}) n^{SR}} \bigg(4c \mathfrak{R}_{n^{SB}}(\ell \circ \mathcal{G}) + 2c  \sqrt{\frac{\log 4 /\delta }{2 n^{SB}}}+ B_2 \sqrt{\frac{\log 8 / \delta}{2\Nsim^{SR}}} \bigg). \nonumber 
\end{align}
\label{thm:multisalad}
\end{restatable}

We make several observations about this bound:
\begin{itemize}
    \item The bound scales in $(n^{SB})^{-1/2}$ and $(\Nsim^{SR})^{-1/2}$, where the former comes from the initial reweighting step while the latter comes from the weighted classification step.
    \item The bound is also dependent on the Rademacher complexities of both classifiers $g$ and $h$ used. 
    \item The bound depends on the difference between the simulation and data distributions through quantities $W$, $B_1, B_2, \eta, \hat{g}_{\max}, g_{\max}$. If the distributions have very different densities, these quantities will all be large, increasing the generalization error.
\end{itemize}

We comment how this bound is different when instantiated for \salad versus \multisalad. The following example shows how \salad with one simulation can result in a large $W$ (and other large constants), while \multisalad with two simulations combined can reduce $W$ in the bound.

\begin{example}
Let $\Psim^1(x | y = 0) = \N(\mu, \sigma^2)$, $\Psim^2(x | y = 0) = \N(-\mu, \sigma^2)$ be Gaussian distributions on $x$ with $\mu, \sigma^2 \in \R$, and let the true background distribution $\p(\cdot | y = 0)$ be a mixture of the Gaussians on $x$, $\p(x | y = 0) = \frac{1}{2}\Psim^1 + \frac{1}{2} \Psim^2$. Let $\Psim^1, \Psim^2,$ and $\p$ have the same marginal distribution over $m$ with $x \independent m | y$. Then, if we only use one simulation $\Psim^1$,
\begin{align*}
    w(x, m) &= \frac{p(x, m | y = 0)}{\psim^1(x, m | y = 0)} = \frac{p(x | y = 0)}{\psim^1(x | y = 0)} \\
    & = \frac{\frac{1}{2\sigma \sqrt{2\pi}} \exp\Big(-\frac{(x - \mu)^2}{2\sigma^2} \Big) + \frac{1}{2\sigma \sqrt{2\pi}} \exp\Big(-\frac{(x + \mu )^2}{2\sigma^2} \Big)}{\frac{1}{\sigma \sqrt{2\pi}} \exp\Big(-\frac{(x - \mu)^2}{2\sigma^2} \Big)} \\
    &= \frac{1}{2} + \frac{1}{2}\exp \bigg(\frac{(x - \mu)^2}{2\sigma^2} - \frac{(x + \mu)^2}{2\sigma^2} \bigg) = \frac{1}{2} + \frac{1}{2} \exp\bigg(\frac{-2x\mu}{\sigma^2} \bigg).
\end{align*}

Therefore, as $x \rightarrow -\infty$, $W \rightarrow \infty$. 
However, if we define $\Psim$ as the distribution of the two simulation datasets concatenated, we have that $\psim(x | y = 0) = p(x | y = 0)$, and as a result, $W \rightarrow 1$, making the generalization error bound smaller.
\end{example}

From this example, we can see that significantly differing simulation and data distributions can result in large, unbounded weight ratios, which are correlated with poor performance.\footnote{The bound in Theorem~\ref{thm:multisalad} is meant to provide a general understanding of SALAD's performance. It can be made tighter by replacing terms that are maxima like $M$ and $B_2$ with terms that are based on the overall data distributions (e.g. variance, as in Ref.~\cite{cortes2010learning}). Variance-based bounds are less likely to be vacuous, but will still demonstrate how performance is dependent on the intrinsic differences between the two distributions.} This concretely motivates our algorithmic objective to combine multiple simulation datasets as to closely approximate the true data.

\subsection{Empirical Results}

\begin{figure}
    \centering
    \includegraphics[width=0.8\textwidth]{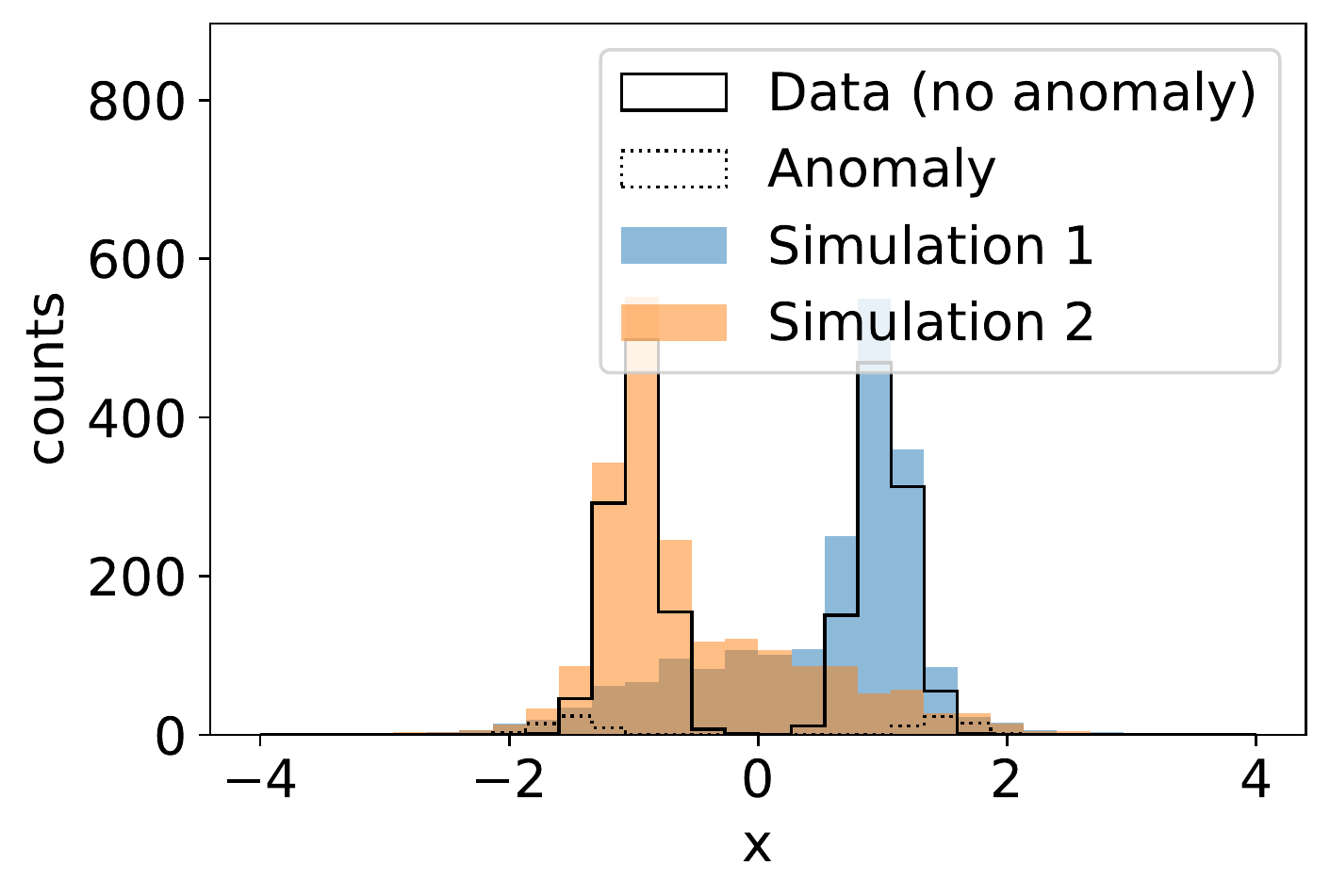}
    \caption{Synthetic data for evaluating Multi-SALAD.}
    \label{fig:salad_data_distr}
\end{figure}

To demonstrate how \multisalad can improve over using only one simulation and over using simulations separately, we consider a synthetic experiment with two simulation datasets\footnote{We find that the differences between the simulations in the LHC Olympics are not enough to see a noticeable gain from \multisalad over \salad.}.
The true background is $\mathcal{P}(\cdot | y = 0) = \frac{1}{2} \N(-1, 0.2) + \frac{1}{2}\N(1, 0.2)$, and the anomaly is $\mathcal{P}(\cdot | y = 1) = \frac{1}{2} \N(-2, 0.2) + \frac{1}{2} \N(2, 0.2)$.
Simulation 1 is $\Psim^{1} = \frac{1}{2} \N(1, 0.2) + \frac{1}{2} \N(0, 1)$, and simulation 2 is $\Psim^{2} = \frac{1}{2}\N(-1, 0.2) + \frac{1}{2}\N(0, 1)$. 
We generate $2000$ points from the true background and $100$ points that are anomalies to form $\D$, and $2000$ points each from $\Psim^1$ and $\Psim^2$ to form $\dsim_1$ and $\dsim_2$.
We construct signal and sideband regions from these by splitting datasets in half randomly, assuming they follow the same distribution over $x$ (i.e., $m$ is independent of $x$) except that there is no anomaly in the sideband regions. 
A visualization is shown in Figure~\ref{fig:salad_data_distr}.

Intuitively, the anomaly is only slightly different from the background data, which makes it important to learn a good reweighting function from the simulations.
Because each simulation alone diverges greatly from the data for one mode, each individual reweighting may not approximate the true $\mathcal{P}(\cdot | y = 1)$ well.
On the other hand, if we combine both simulation datasets together, the aggregate distribution has smaller weights with lower variance, which can allow for more accurate reweighting. This is demonstrated in Figure~\ref{fig:salad_reweigh_sb}, which depicts the reweighting in the sideband region. 
Figure~\ref{fig:salad_reweigh_sr} depicts the reweighting's interpolation into the signal region, where we introduce an additional baseline \switchsalad, which uses $k$ separate weight functions $w_i(x, m)$ and switches among them in the reweighted loss function $L_S$.
In all but the bottom right subfigure in both figures, the reweighted simulation data poorly approximates the true background data. As a result, a classifier trained to distinguish between the high-variance reweighted simulation and the true background data plus some small anomaly will more likely learn the distinctions coming from poor approximation, rather than anomaly. In particular, note that \switchsalad results in significant overweighting in Figure~\ref{fig:salad_reweigh_sr}.

\begin{figure}[t]
    \centering
    \includegraphics[width=0.45\textwidth]{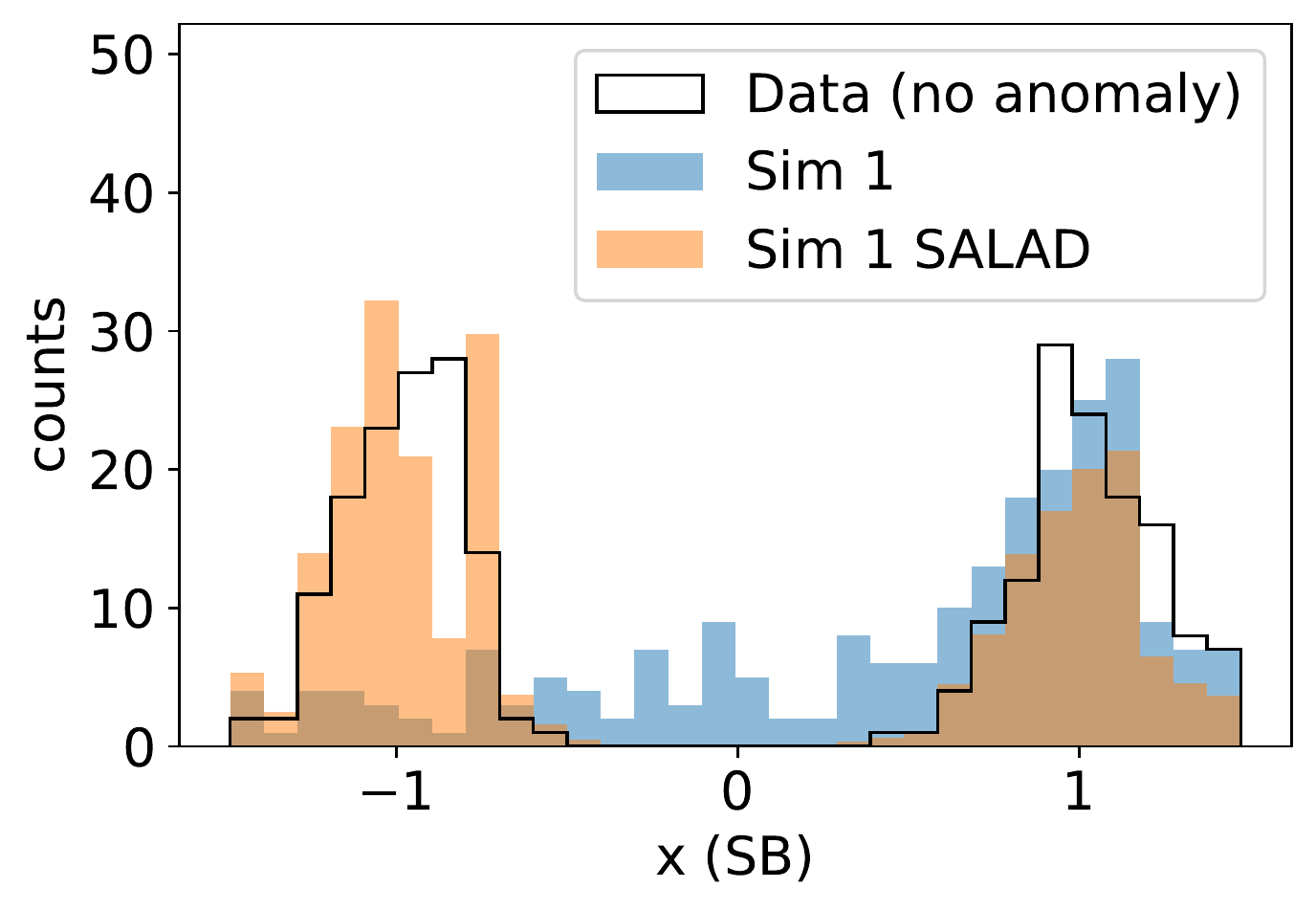}
    \includegraphics[width=0.45\textwidth]{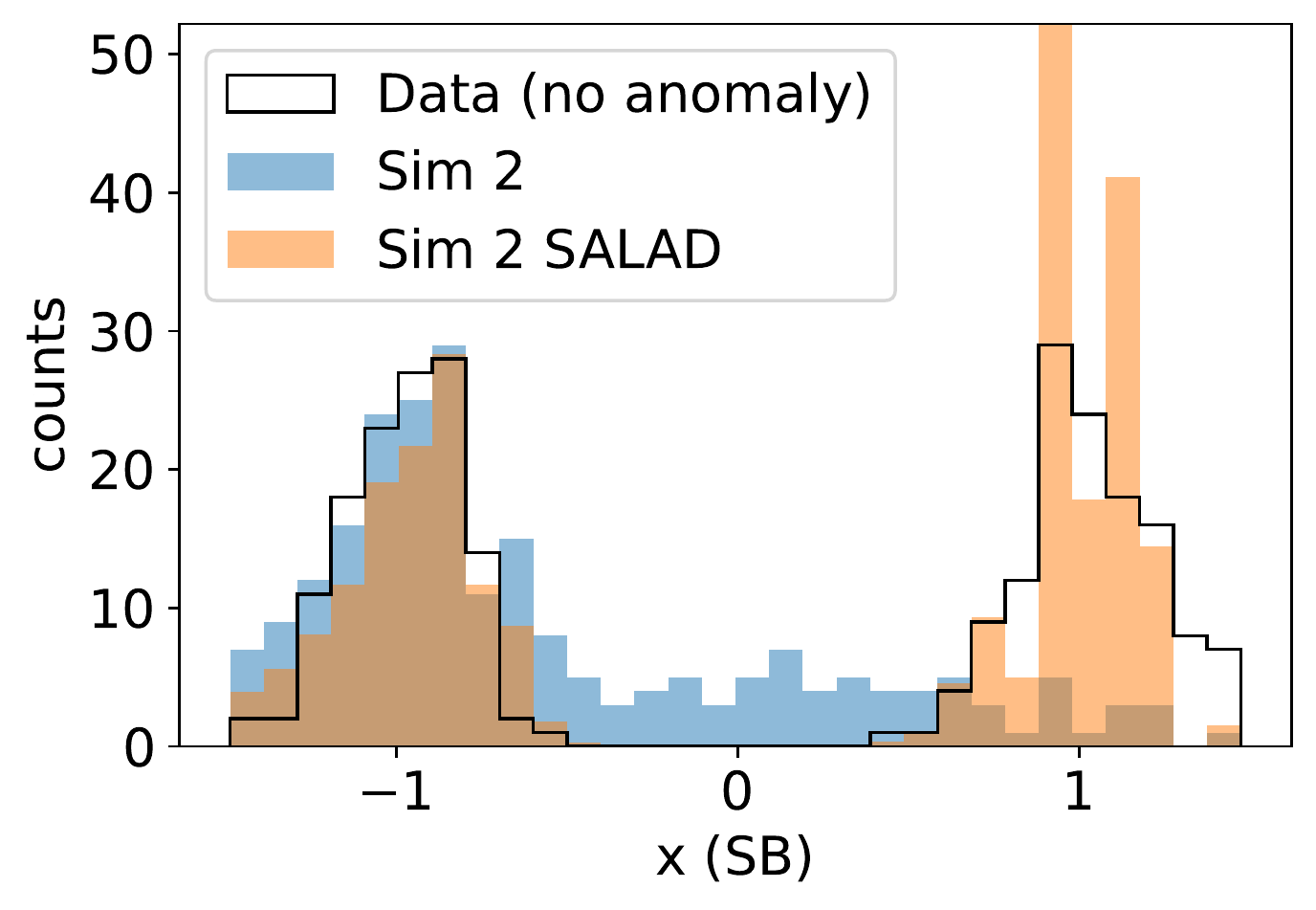}
    \includegraphics[width=0.45\textwidth]{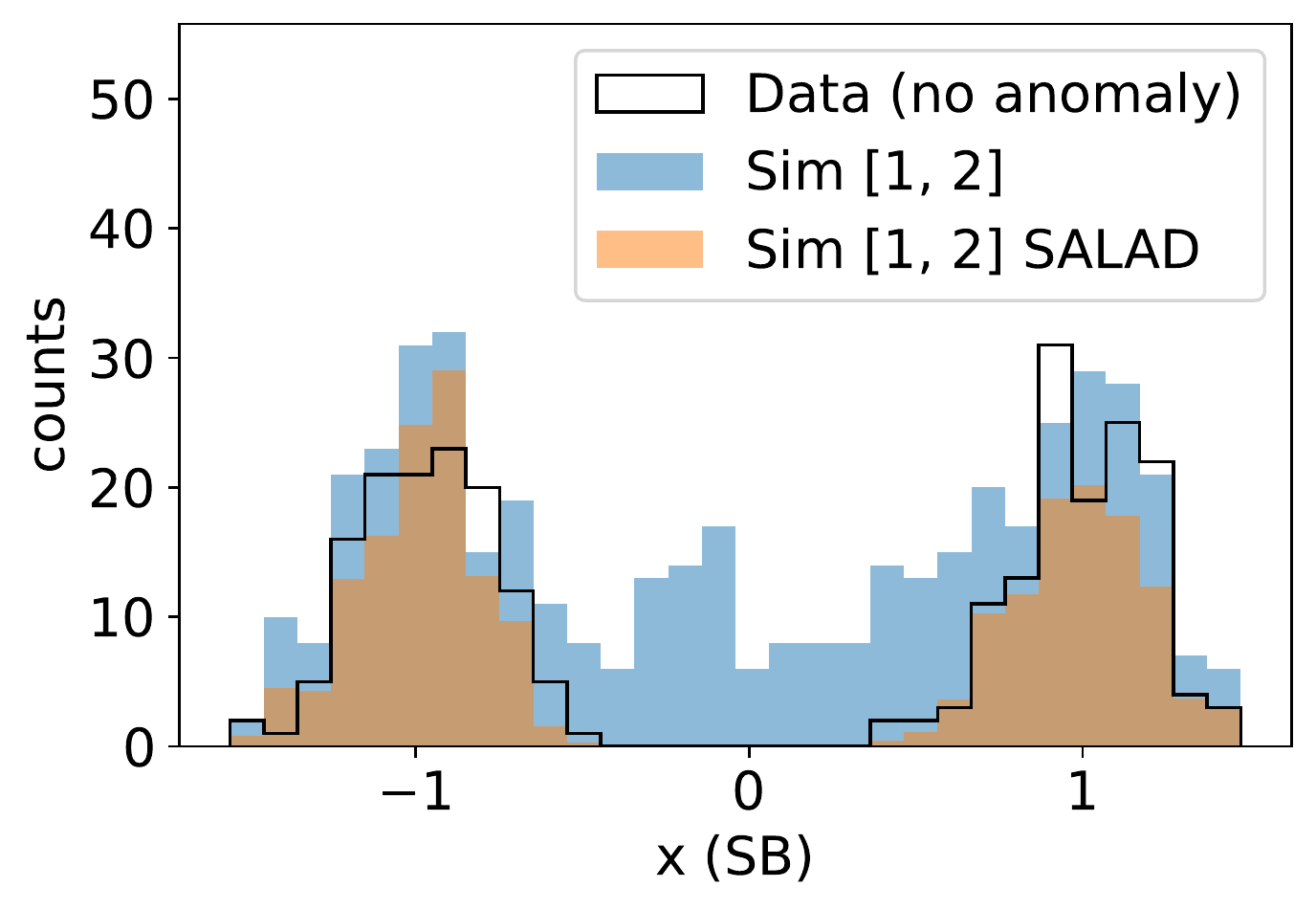}

    \caption{Top left: SALAD reweighting using simulation 1 on sideband region. Top right: reweighting using simulation 2. Bottom: reweighting using simulation 1 and 2 combined.}
    \label{fig:salad_reweigh_sb}
\end{figure}

\begin{figure}[t]
    \centering
    \includegraphics[width=0.45\textwidth]{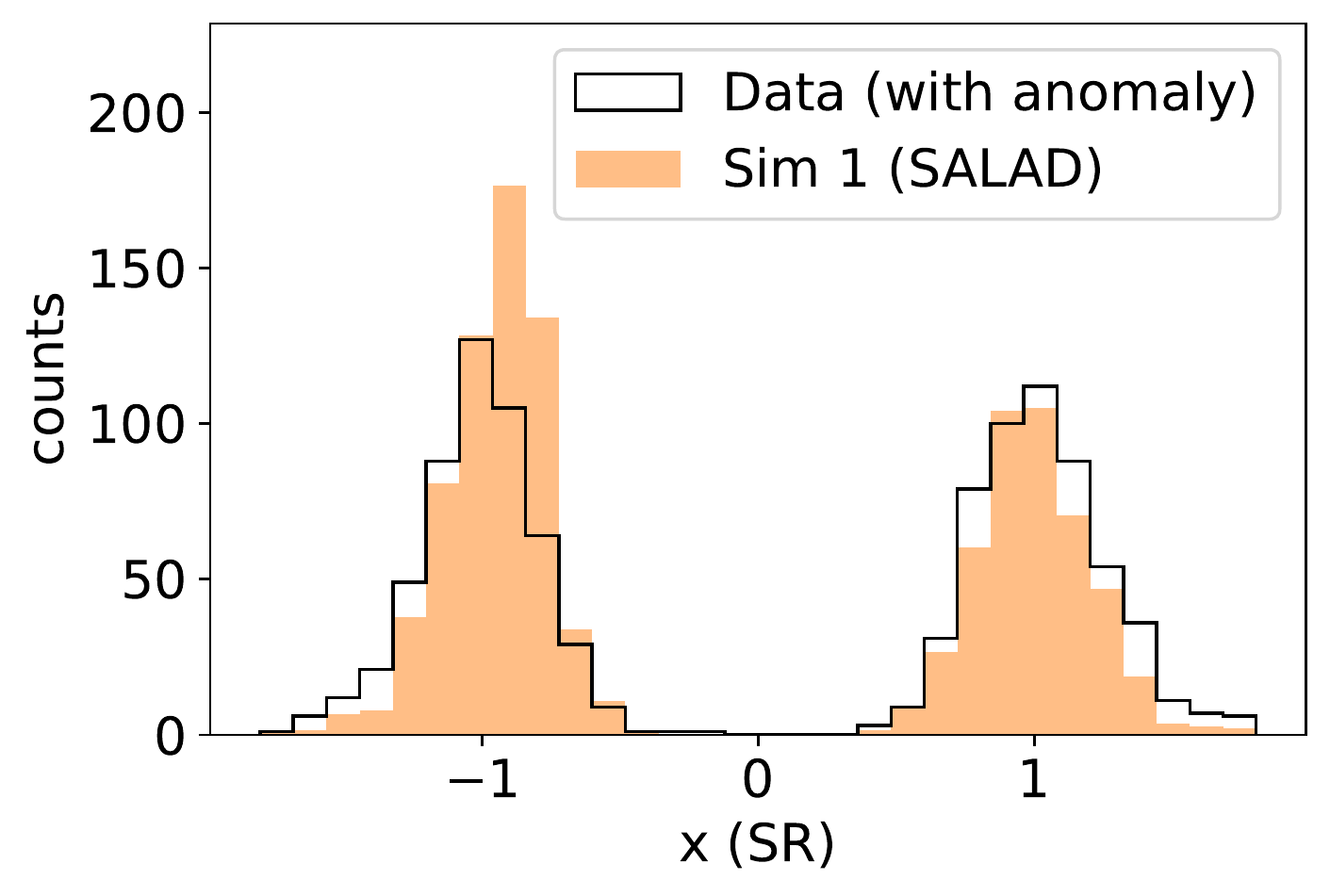}
    \includegraphics[width=0.45\textwidth]{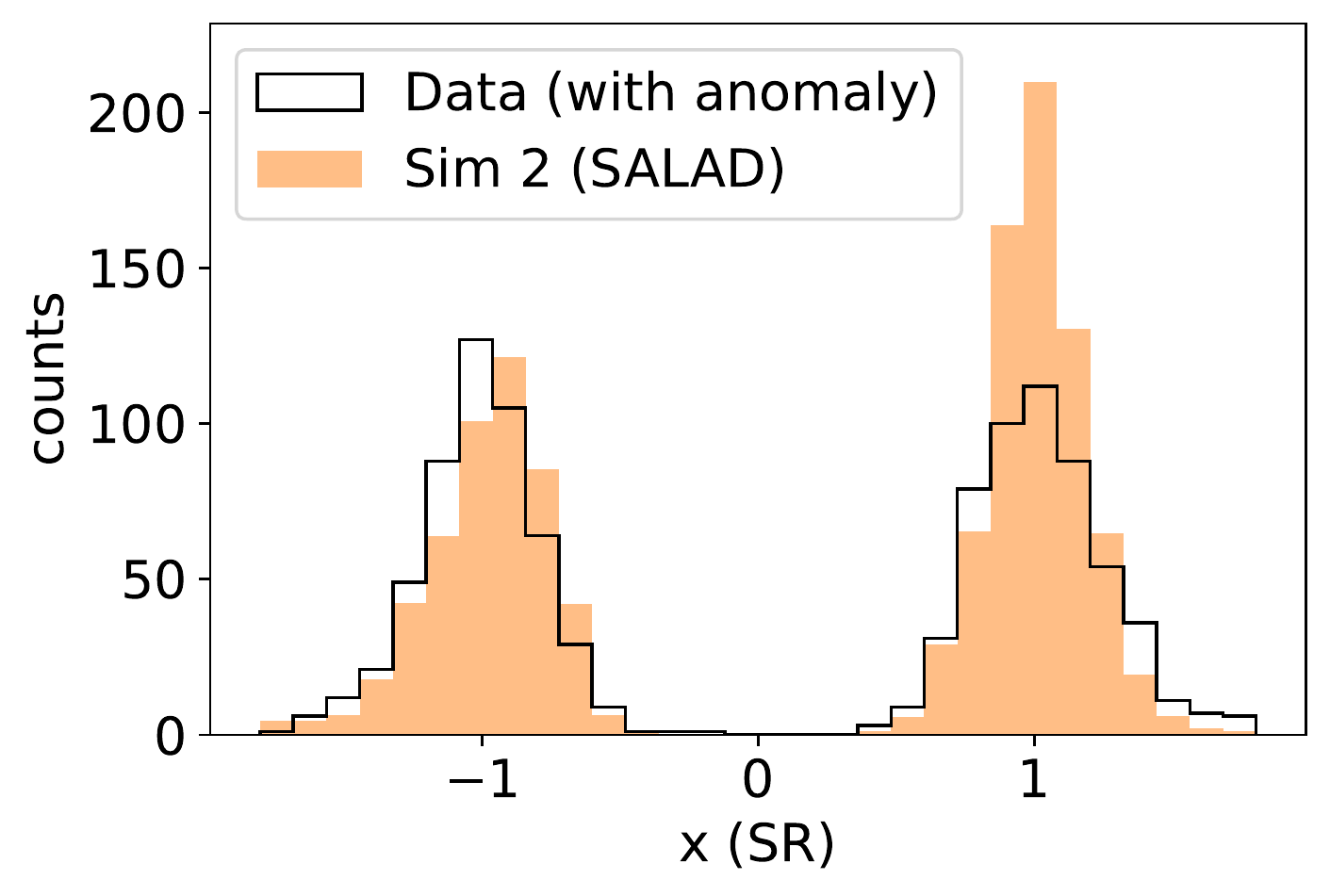}
    \includegraphics[width=0.45\textwidth]{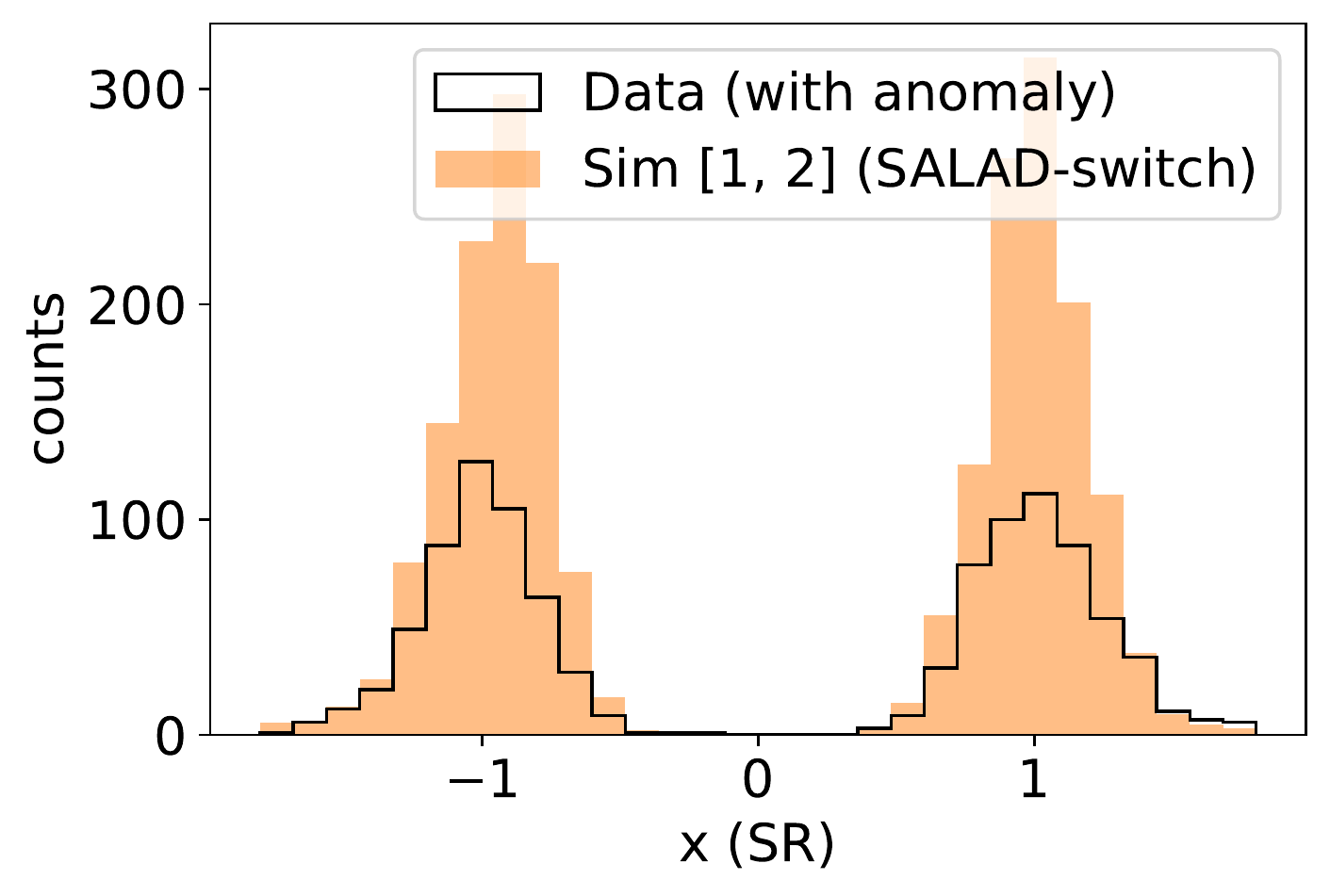}
    \includegraphics[width=0.45\textwidth]{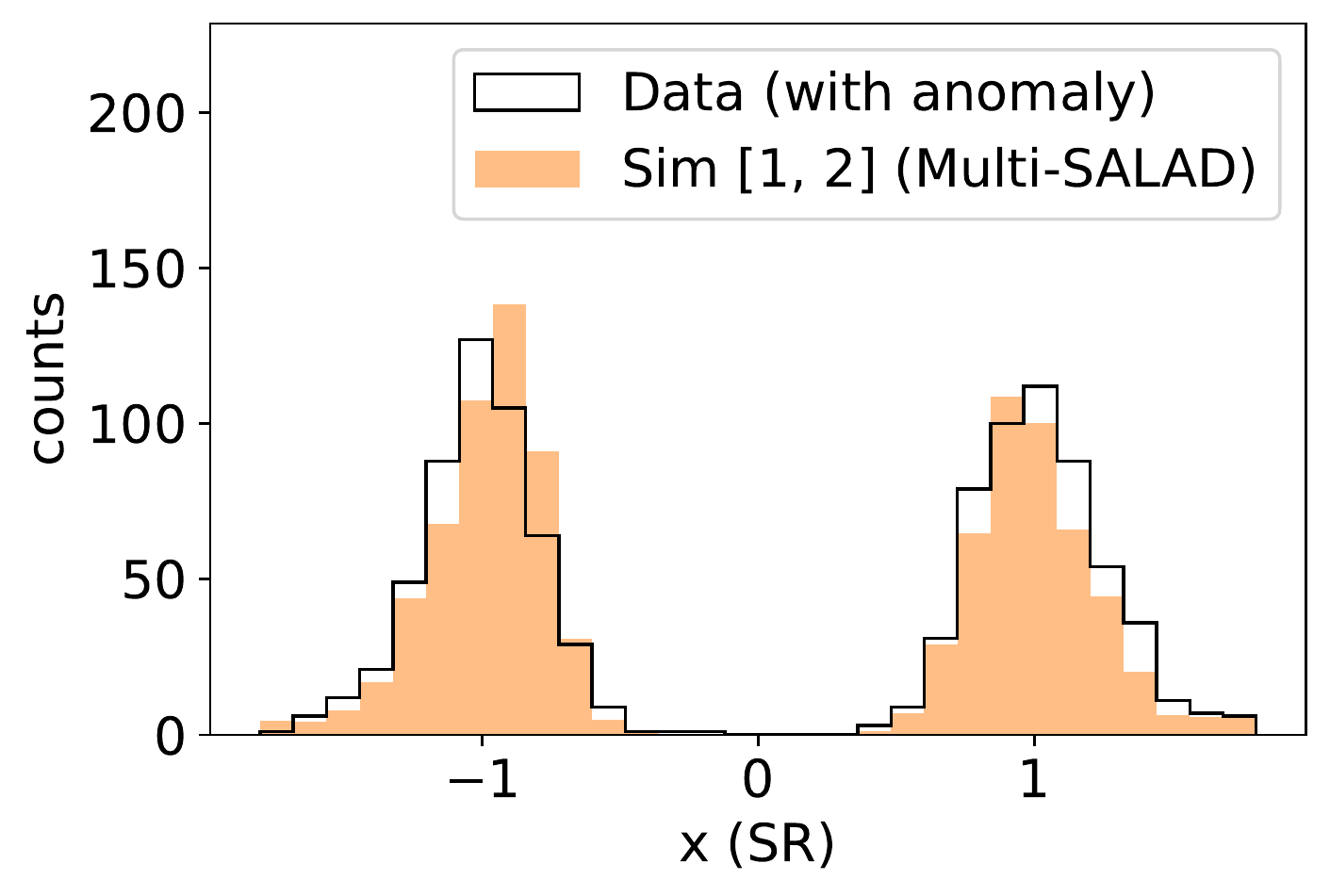}

    \caption{Top left: SALAD reweighting using simulation 1 on signal region. Top right: reweighting using simulation 2. Bottom left: using both simulation 1 and 2 weights separately. Bottom right: reweighting using simulation 1 and 2 combined.}
    \label{fig:salad_reweigh_sr}
\end{figure}

With these observations, we present the signal efficiency to rejection rate of each method in Figure~\ref{fig:salad_curve}, where we compare \multisalad against \salad using simulation 1 only, SALAD using simulation 2 only, and \switchsalad. Table~\ref{tab:salad} contains the accuracy and AUC scores for each method. Averaged over $10$ random seeds, \multisalad outperforms other methods. The signal efficiency to rejection rate for each of the 10 runs is available in Appendix~\ref{supp:exp}.

\begin{figure}[t]
    \centering
    \includegraphics[width=0.75\textwidth]{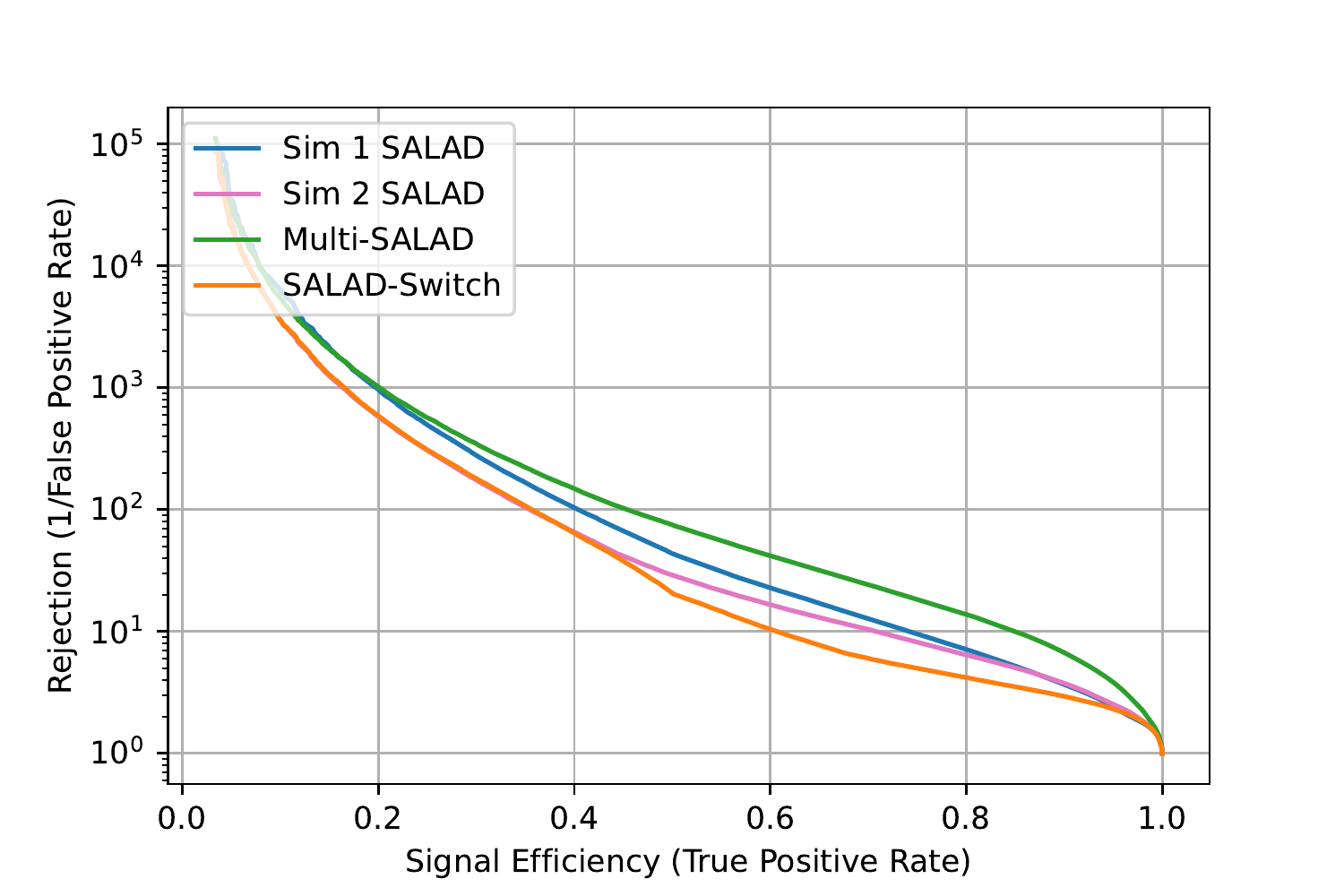}
    \caption{Signal efficiency to rejection of Multi-SALAD versus other baselines (weighted and unweighted).}
    \label{fig:salad_curve}
\end{figure}

\begin{table*}[h!]
    \centering
    \scriptsize
    \begin{tabular}{| c|c | c| c| c| c| c| c|} \hline
          &  \multicolumn{2}{|c|}{Simulation 1} & \multicolumn{2}{|c|}{Simulation 2} & \multicolumn{3}{|c|}{Simulation 1 and 2} \\ \hline
         Method & None & \salad & None & \salad & None & \switchsalad & \multisalad \\ \hline 
         Accuracy & $43.8_{\pm 2.2}$ & $62.5_{\pm 8.8}$ & $42.7_{\pm 3.6}$ & $64.3_{\pm 12.3}$ & $50.0_{\pm 0.0}$ & $54.3_{\pm 6.2}$ & $\mathbf{64.8}_{\pm 9.3}$ \\ \hline 
         AUC & $28.5_{\pm 4.2}$ & $80.7_{\pm 14.5}$ & $27.4_{\pm 4.5}$ & $78.7_{\pm 18.2}$ & $15.4_{\pm 5.3}$ & $74.7_{\pm 17.0}$ & $\mathbf{90.8}_{\pm 10.2}$ \\ \hline
    \end{tabular}
    \caption{Accuracy and AUC scores ($\%$) for \multisalad on two simulation datasets. We compare to \switchsalad (different reweighting), as well as standard \salad on individual simulations and no reweighting. Performance is averaged over $10$ random runs with one standard deviation reported.}
    \label{tab:salad}
\end{table*}

\section{Conclusions and Outlook}
\label{sec:conclusion}

We extend two resonant AD approaches to incorporate multiple reference datasets. For \multicwola, we draw from weak supervision models to handle multiple resonant features. For \multisalad, we combine multiple simulation datasets to best approximate the background process. Future work includes 1) exploring \multisalad's applicability on real data and algorithms for sampling from simulation datasets 2) extending \multicwola to model more complex relationships among resonant features and 3) using such approaches together over multiple simulations and resonant features, effectively utilizing as much information as possible. 


\section*{\label{sec::acknowledgments}Acknowledgments}

We thank David Shih and Jesse Thaler for useful discussions and comments about the manuscript.  BN was supported by the Department of Energy, Office of Science under contract number DE-AC02-05CH11231.
FS is grateful for the support of the NSF under CCF2106707 and the Wisconsin Alumni Research Foundation (WARF).
We gratefully acknowledge the support of NIH under No. U54EB020405 (Mobilize), NSF under Nos. CCF1763315 (Beyond Sparsity), CCF1563078 (Volume to Velocity), and 1937301 (RTML); ARL under No. W911NF-21-2-0251 (Interactive Human-AI Teaming); ONR under No. N000141712266 (Unifying Weak Supervision); ONR N00014-20-1-2480: Understanding and Applying Non-Euclidean Geometry in Machine Learning; N000142012275 (NEPTUNE); NXP, Xilinx, LETI-CEA, Intel, IBM, Microsoft, NEC, Toshiba, TSMC, ARM, Hitachi, BASF, Accenture, Ericsson, Qualcomm, Analog Devices, Google Cloud, Salesforce, Total, the HAI-GCP Cloud Credits for Research program,  the Stanford Data Science Initiative (SDSI), and members of the Stanford DAWN project: Facebook, Google, and VMWare. The U.S. Government is authorized to reproduce and distribute reprints for Governmental purposes notwithstanding any copyright notation thereon. Any opinions, findings, and conclusions or recommendations expressed in this material are those of the authors and do not necessarily reflect the views, policies, or endorsements, either expressed or implied, of NIH, ONR, or the U.S. Government.

\bibliographystyle{JHEP.bst}
\bibliography{main,HEPML}

\newpage 

\appendix 

\section*{Appendix}
We provide a glossary of notation in~\ref{supp:glossary}. We provide algorithmic details for \multisalad in Section~\ref{supp:algs}. We present additional theoretical results on Rademacher complexities and the asymptotic behavior of \salad in Section~\ref{supp:more_theory}. In section~\ref{supp:proofs}, we provide proofs for our theoretical results. In section~\ref{supp:exp}, we provide additional experimental details.

\section{Glossary} \label{supp:glossary}

The glossary is given in Table~\ref{table:glossary}.

\begin{table*}[ht]
\centering
\small
\begin{tabular}{l l}
\toprule
Symbol & Used for \\
\midrule
$x$ & Discriminative feature $x \in \X$. \\
$m$ & Resonant feature vector of length $k$, $m = [m^1, \dots, m^k] \in \R^k$. \\
$y$ & True unknown label $y \in \Y = \{0, +1\}$, where $0$ is background and $1$ is signal. \\
$\mathcal{P}, p$ & Distribution and density of data $(x, m, y)$. \\
$I_{m^i}$ & Interval along which $i$th resonant feature $m^i$ is thresholded to produce \\
& signal region and sideband. \\
$SR, SB$ & Signal region and sideband. For an interval $I_{m^i}$, $SR_i = \{(x, m): m^i \in I_{m^i}\}$   \\
& and $SB_i = \{(x, m): m^i \notin I_{m^i}\}$. \\
$f$ & Classifier $f: \X \rightarrow \Y$ used for anomaly detection. \\
$\D$ & Unlabeled dataset $\D = \{(x_i, m_i)\}_{i = 1}^n$ of discriminative and resonant features. \\
$\D_{SR}, \D_{SB}$ & Signal region and sideband of $\D$, $\D_{SR} = \D \cap SR$, $\D_{SB} = \D \cap SB$. \\
$\eta_{SR}, \eta_{SB}$ & Mixture weights corresponding to $p(y = 1 | x \in SR)$ and $p(y = 1 | x \in SB)$. \\
& It is assumed that $\eta_{SR} > \eta_{SB}$. \\
$M_i(m)$ & Noisy membership label for the $i$th resonant feature, equal to $0$ if $x \in \D_{SB_i}$ \\
& and $1$ if $x \in \D_{SR_i}$. $\mathbf{M}(m) = M_1(m), \dots, M_k(m)$.\\
$\hat{y}$ & Weak label drawn from estimated distribution on $p(y | \mathbf{M}(m))$. \\
$\theta_y, \theta_i$ & Canonical parameters of graphical model on $y, \mathbf{M}(m)$ in~\eqref{eq:pgm}. \\
& $\theta_y$ scales with the class balance of $y$ and $\theta_i$ scales with the accuracy of $M_i(m)$. \\
$Z$ & Partition function used for normalizing distribution $p(y, \mathbf{M}(m))$ in~\eqref{eq:pgm}. \\
$\widetilde{y}, \widetilde{\mathbf{M}}(m)$ & $y$ and $\mathbf{M}(m)$ scaled from $\{0, 1\}$ to $\{-1, 1\}$. \\
$\alpha_i$ & Accuracy parameter $\alpha_i = p(M_i(m) = 1 | y = 1)$ for the membership label \\
& of the $i$th resonant feature. \\
$\ell_C$ & Loss function $\ell_C: \Y \times \Y \rightarrow \R$ for training classifier $f$. \\
$L_C(f)$ & Expected loss on labeled data using $f$, $L_C(f) = \E{}{\ell_C(f(x), y)}$. \\
$f^\star$ & Optimal classifier trained on infinite labeled data, $f^\star = \argmin{f \in \F}{ L_C(f)}$. \\
$\hat{L}_C(f)$ & Empirical loss on $\D$ with weak labels using $f$, $\hat{L}_C(f) = \frac{1}{n} \sum_{i = 1}^n \ell_C(f(x_i), \hat{y}_i)$. \\
$\hat{f}$ & Classifier learned using \multicwola, $\hat{f} = \argmin{f \in \F}{\hat{L}_C(f)}$. \\
$\dsim$ & Simulation dataset used in standard \salad, $\dsim = \{(x_i, m_i)\}_{i = 1}^{\Nsim}$. \\
& Has distribution $\Psim$ and density $\psim(\cdot)$. \\
$\dsim_{SB}$, $\dsim_{SR}$ & $\dsim_{SB}=\dsim \cap SB$, $\dsim_{SR} = \dsim \cap SR$. \\
$w(x, m)$ & Density ratio between $\dsim_{SB}$ and $\D_{SB}$ used for reweighting, \\
&$w(x, m) = \frac{p(x, m | y = 0)}{\psim(x, m | y = 0)}$. \\
$\hat{g}$ & Classifier trained to classify $\dsim_{SB}$ vs $\D_{SB}$, used for approximating $w(x, m)$  \\
& when $|\dsim_{SB}| = |\D_{SB}|$. \\
$L_S(h, w)$ & Cross-entropy loss function used to classify $\dsim_{SR}$ reweighted with $w$ vs $\D_{SR}$. \\
$\hat{h}$ & Classifier trained using $L_S$. \\
$\dsim_1, \dots, \dsim_k$ & $k$ multiple simulation datasets used in $\multisalad$. \\
$\widetilde{\mathcal{D}}_{\mathrm{sim}}$ & Dataset aggregated from $\dsim_1, \dots, \dsim_k$. \\
$n^{SR}$ & $n^{SR} = |\D_{SR} |$. \\
$n^{SB}$ & $n^{SB} = |\D_{SB} |$. \\
$\Nsim^{SR}$ & $\Nsim^{SR} = |\dsim_{SR} |$. \\
$h^\star$ & The optimal classifier $h^\star = \argmin{h \in \mathcal{H}}{L_S(h, w)}$. \\
$W$ & The maximum ratio between the simulation and true background, \\
& $W = \max_{x, m} w(x, m)$. \\
\toprule
\end{tabular}
\caption{
	Glossary of variables and symbols used in this paper.
}
\label{table:glossary}
\end{table*}


\section{Additional Algorithmic Details}\label{supp:algs}

\subsection{\multisalad Algorithm} \label{supp:multisalad}

\multisalad is described in Algorithm~\ref{alg:multisalad}.
We have simulation datasets $\dsim_1, \dots \dsim_k$, where $\dsim_i = \{(x_j, m_j)\}_{j = 1}^{\Nsim}$ and all points belong to the background $(y = 0)$. 
As discussed in Section~\ref{sec:salad}, we propose using these simulation datasets by aggregating them into a single simulation dataset $\dsim$ (whether it be with uniform or stratified sampling, etc.) Then the rest of this section proceeds as follows and is a review of the standard \salad method.

\paragraph{Reweighting} First, we learn weights to correct for the bias of the simulated background data. We split the both simulation and true data along $m$ to produce sets $\dsim_{SR}, \dsim_{SB}$ and $\D_{SR}$ and $\D_{SB}$. We train a classifier over $\dsim_{SB}$ and $\D_{SB}$ to distinguish between simulation and real data in the sideband region. That is, we train a binary classifier $\hat{g}$ over points $(x, m, z)$ in the sideband where $x, m$ is either from $\psim(\cdot | y = 0)$ ($z = 0$) or $p(\cdot | y = 0)$ ($z = 1$), where we recall that simulation data only contains $y = 0$, and no anomalies are present in the sideband. Denote $q$ as the joint density of $(x, m, z)$. We define the weight as the estimated likelihood ratio
\begin{align*}
    \hat{w}(x, m) &= \frac{\hat{g}(x, m)}{1 - \hat{g}(x, m)} \approx \frac{q(z = 1 | x, m)}{q(z = 0| x, m)} = \frac{q(x, m | z = 1)}{q(x, m | z = 0)} \cdot \frac{q(z = 1)}{q(z = 0)} \\
    &= \frac{q(x, m | z = 1)}{q(x, m | z = 0)} = \frac{p(x, m | y = 0)}{\psim(x, m | y = 0)}.
    \label{eq:w}
\end{align*}

Here, we assume that $q(z = 1) = q(z = 0)$ (i.e. balanced simulation and real dataset, which we can always ensure by generating more or less simulation data). Equality is obtained in the expression above when $\hat{g}$ is Bayes-optimal. 

\paragraph{Training} The above $\hat{w}(x, m)$ is defined on the sideband region. Next, we interpolate and correct the bias of the simulation in the signal region.
Let $\dsim_{SR}$ be the set of simulation data in the signal region of size $\Nsim^{SR}$, and let $\D_{SR}$ be the set of true data in the signal region of size $\Ndata^{SR}$, for a total of $n^{SR}$ points.
We train a classifier $h$ to distinguish between the reweighted simulated data, which approximates true background data, and the true data. In particular, the loss function used is 
\begin{align}
    \hat{L}_S(h, \hat{w}) = - \frac{1}{n^{SR}} \bigg(\sum_{x \in \D_{SR}} \log h(x, m) + \sum_{x \in \dsim_{SR}} \hat{w}(x, m) \log (1 - h(x,m )) \bigg).
\end{align}

In expectation with an optimal $w$, we can see that minimizing this loss is equivalent to minimizing the cross-entropy loss on a task that distinguishes between points drawn from $p$ and points drawn from $p(\cdot | y =0)$ in the signal region. Therefore, $h$ can be used for anomaly detection. The procedure is summarized in Algorithm~\ref{alg:multisalad}.

\begin{algorithm}[t]
\caption{\multisalad}
\begin{algorithmic}[1]
\STATE \textbf{Input:} Simulation datasets $\dsim_1, \dots, \dsim_k$ and real dataset $\D$.
\STATE Construct overall simulation dataset $\dsim = \bigcup_{i = 1}^k \dsim_i$.
\STATE Split each dataset into signal region and sideband region using resonant feature $m$ to get $\{\dsim_{SR}, \dsim_{SB} \}$ and $\{\D_{SR}, \D_{SB} \}$.
\STATE Learn weight $\hat{w}(x, m) = \frac{\hat{g}(x, m)}{1 - \hat{g}(x, m)}$, where $\hat{g}$ is a classifier that distinguishes data $\D_{SB}$ from simulation $\dsim_{SB}$ in the sideband region.
\STATE Train a new classifier $\hat{h}$ on the signal region to distinguish between points in $\D_{SR}$ and points in $\dsim_{SR}$ reweighted by $\hat{w}$, using the following loss:
\begin{align}
    \hat{L}_S(h, \hat{w}) = - \frac{1}{n^{SR}} \bigg(\sum_{x \in \D_{SR}} \log h(x, m) + \sum_{x \in \dsim_{SR}} \hat{w}(x, m) \log (1 - h(x,m )) \bigg).
\end{align}
\STATE \textbf{Output:} Classifier output $\hat{h}(x, m)$, which yields a score that is thresholded for anomaly detection.
\end{algorithmic}
\label{alg:multisalad}
\end{algorithm}

\section{Additional Theoretical Results} \label{supp:more_theory}
\subsection{The Need for 3 Resonant Features}
\label{subsec:three}
We show that to identify the model \eqref{eq:pgm}, we need at least $k = 3$ resonant features.

\begin{lemma}
If $k =1$ or $k=2$ in model~\eqref{eq:pgm}, the parameters $\theta_1$ and $\theta_2$ cannot be recovered from the observable quantities.
\label{lemma:small_k}
\end{lemma}
\begin{proof}
The strategy we use to show that the model cannot be identified for $k=1$ or $k=2$ is to prove that the observable distributions $P(\widetilde{M}_1(m), \ldots, \widetilde{M}_k(m))$ are consistent with multiple values of $\theta$. We do so by direct calculation. 

First, consider the case of $k=1$. Set $\theta_y = 0$ for simplicity. Then, the model is $\frac{1}{Z}\exp(\theta \widetilde{M}_1(m) \widetilde{y})$. Then $Z = 2\exp(\theta) + 2\exp(-\theta)$, and \[P(\widetilde{M}_1(m)= 1) = \frac{\exp(\theta) + \exp(-\theta)}{2\exp(\theta) + 2\exp(-\theta)} = \frac{1}{2}.\]
Thus, any $\theta$ value produces the same observable distribution, so that we cannot identify $\theta$.

Next, we consider $k=2$. Again, set $\theta_y = 0$. The model is now $\frac{1}{Z}\exp(\theta_1 \widetilde{M}_1(m) \widetilde{y} + \theta_2 \widetilde{M}_2(m) \widetilde{y})$. We similarly compute 
\[Z = 2(\exp(\theta_1 + \theta_2) + \exp(-\theta_1 + \theta_2) + \exp(\theta_1 - \theta_2) + \exp(-\theta_1-\theta_2)).
\]
The observable distribution is now $P(\widetilde{M}_1(m), \widetilde{M}_2(m))$. We have that
\[P(\widetilde{M}_1(m) = 1, \widetilde{M}_2(m) = 1) = \frac{1}{Z} (\exp(\theta_1 + \theta_2) + \exp(-\theta_1- \theta_2)), \]
and 
\[P(\widetilde{M}_1(m) = 1, \widetilde{M}_2(m) = -1) = \frac{1}{Z} (\exp(\theta_1 - \theta_2) + \exp(-\theta_1 + \theta_2)).\]

Note that we have $P(\widetilde{M}_1(m) = -1, \widetilde{M}_2(m) = -1) = P(\widetilde{M}_1(m) = 1, \widetilde{M}_2(m) = 1)$ and $P(\widetilde{M}_1(m) = -1, \widetilde{M}_2(m) = 1) = P(\widetilde{M}_1(m) = 1, \widetilde{M}_2(m) = -1)$.

As a result, we have the same distribution $P(\widetilde{M}_1(m), \widetilde{M}_2(m))$ for the parameters $\theta_1, \theta_2 = a, b$ and for $\theta_1, \theta_2 = b, a$, where $a, b$ are some non-negative values. If $a \neq b$, we end up with at least two solutions that cannot be distinguished, completing the proof.
\end{proof}

\subsection{Rademacher Complexity Bounds} \label{supp:rademacher}

We present bounds on the Rademacher complexity $\mathfrak{R}_n(\F)$ of various models $\F$.
For all of the $\F$ below, we obtain $\mathfrak{R}_n(\ell \circ \F)$ by computing $\mathfrak{R}_n(\F)$. These two Rademacher complexities are equal when we assume that $\ell$ is $1$-Lipschitz and apply Talagrand's lemma.

\begin{itemize}
    \item \textbf{Linear models:} We define $f_{\theta}(x) = \theta^\top x$ with $\|\theta\|_2 \leq B$ and $E[\|x\|_2^2] \leq C^2$, $\mathfrak{R}_n(\F) \leq \frac{BC}{\sqrt{n}}$ \cite[Theorem 5.5]{Ma22}. 

    \item \textbf{Two-layer feed-forward neural networks (MLPs):} We define $f_{\theta}(x)$ where $\theta = (U, w)$ are the parameters for the weights for the two layers of an MLP. Here $U \in \mathbb{R}^{m \times d}$ and $w \in \mathbb{R}^m$. Suppose ReLU is the activation function, $\|w\|_2 \leq B_w$, $\|u_i\|_2 \leq B_u$ for all $1 \leq i \leq m$, and that $\mathbb{E}[\|x\|_2^2 \leq C^2$. Then, $\mathfrak{R}_n(\F) \leq 2B_wB_uC\sqrt{\frac{m}{n}}$ \cite[Theorem 5.9]{Ma22}.
    
    \item \textbf{Kernels:} Let $k : \mathcal{X} \times \mathcal{X} \rightarrow \mathbb{R}$ be a continuous symmetric function so that for $x_1, \ldots, x_n$, the matrix given by $K_{ij} = k(x_i, x_j)$ is positive semidefinite. The class of kernel estimators consists of functions $f(x) = \sum_{i=1}^n \alpha_i k(X_i, x)$. Suppose that $\sum_{i,j} \alpha_i \alpha_j k(X_i, X_j) \leq B^2$; then, from~\cite{bartlett2002rademacher}, $\mathfrak{R}_n(\F) \leq 2B \sqrt{\frac{E[k(X,X)]}{n}}$.
    For particular kernels it is easy to bound the term in the numerator above. For example, we consider the RBF kernel which has maximum one, yielding $\mathfrak{R}_n(\F) \le \frac{2 B}{\sqrt{n}}$. 
\end{itemize}

\subsection{Asymptotic behavior of \salad's $\hat{L}_S(h, w)$} \label{supp:salad_loss}

\begin{restatable}[]{lemma}{saladloss}
Assume that the reweighting function is Bayes-optimal, meaning that $\hat{w}(x, m) = w(x, m)$. Then, 
$$\lim_{n^{SR} \rightarrow \infty} \hat{L}(h, \hat{w}) \propto L_{CE}(h),$$ 
where $L_{CE}(h) = \E{x,m, z' = 1}{ -\log h(x, m)} + \E{x, m, z' = 0}{-\log (1 - h(x, m))}$ is the cross entropy loss on label $z' = \begin{cases} 1 & x,m \sim \p \\ 0 & x,m \sim p(\cdot | y = 0) \end{cases}$.
\label{lemma:salad_loss}
\end{restatable}

\begin{proof}

Let $\Ndata^{SR}$ be the number of points from $\D$ that belong to the signal region.
Under our assumptions, the empirical loss function can be written as
\begin{align*}
    \hat{L}(h, \hat{w}) \propto &- \frac{\Ndata^{SR}}{n^{SR}} \cdot \frac{1}{\Ndata^{SR}} \;\;\; \quad \sum_{\mathclap{x \in \D_{SR}}} \log h(x, m) \\
    &- \frac{\Nsim^{SR}}{n^{SR}} \cdot \frac{1}{\Nsim^{SR}} \;\; \quad \sum_{\mathclap{x \in \dsim_{SR}}} \qquad \frac{p(x, m | y = 0)}{\psim(x, m | y = 0)} \log (1 - h(x, m)).
\end{align*}

As $n^{SR} \rightarrow \infty$, the first term approaches $-\Pr(z' = 1) \cdot \E{x, m \sim \p}{\log h(x, m)} = -\Pr(z' = 1) \cdot \E{x, m | z' = 1}{\log h(x, m)}$. For the second term, we can construct $\Ndata^{SR, 0}$, the amount of data where $x$ is from $p(\cdot | y = 0)$, to be equal to $\Nsim^{SR}$ such that the expression asymptotically approaches $-\Pr(z' = 0) \cdot \E{x, m \sim \Psim}{\frac{p(x, m | y = 0)}{\psim(x, m | y = 0)} \log (1 - h(x, m))}$. Performing a change of expectation, this is equal to $-\Pr(z' = 0) \cdot \E{x, m | z' = 0}{ \log (1 - h(x, m))}$.
Putting this together, we have that
\begin{align*}
    \lim_{\mathclap{n^{SR} \rightarrow \infty}} \hat{L}(h, \hat{w}) &\propto -\Pr(z' = 1) \E{x, m | z' = 1}{\log h(x, m)} - \Pr(z' = 0) \E{x, m | z' = 0}{\log (1 - h(x, m))}  \\
    &= L_{CE}(h).
\end{align*}

\end{proof}

\section{Proofs} \label{supp:proofs}

\subsection{Proof of Theorem~\ref{thm:multicwola}}

\begin{proof}
From Theorem 3 of~\cite{Fu20}, we have that $L_C(\hat{f}) - L_C(f^\star)$ is bounded by the traditional ERM generalization gap of $L_C(\bar{f}) - L_C(f^\star)$, where $\bar{f} = \argmin{f \in \F}{\frac{1}{n}\sum_{i = 1}^m \ell(f(x_i, m_i), y_i)}$ is the classifier learned on labeled data, plus the term  $ \frac{c_1}{e_{\min} a_{\min}^5} \Big(\sqrt{\frac{k}{n}} + \frac{c_2 k}{\sqrt{n}} \Big)$.

We can apply standard learning theory bounds on $L_C\bar{f}) - L_C(f^\star)$. In particular, this quantity is equal to
\begin{align*}
    L_C(\bar{f}) - L_C(f^\star) &= (L_C(\bar{f}) - \hat{L}_C(\bar{f})) + (\hat{L}_C(\bar{f}) - \hat{L}_C(f^\star)) + (\hat{L}_C(f^\star) - L_C(f^\star)) \\
    &\le L_C(\bar{f}) - \hat{L}_C(\bar{f}) + \hat{L}_C(f^\star) - L_C(f^\star) \\
    &\le 2\sup_{f \in \F} |L_C(f) - \hat{L}_C(f) |, 
\end{align*}
where we have used the fact that $\hat{L}_C(\bar{f}) \le \hat{L}_C(f^\star)$. Then, using uniform convergence bounds, such as Theorem 3.3 of~\cite{mohri2018foundations}, we have
\begin{align*}
    L_C(\bar{f}) - L_C(f^\star) \le 2 \bigg(2 \mathfrak{R}_n (\ell \circ \F) + \sqrt{\frac{\log 2 / \delta}{2n}} \bigg).
\end{align*}

This gives us our desired result. 

\end{proof}

\subsection{Proof of Theorem~\ref{thm:multisalad}}
\begin{proof}

We define the true (cross-entropy) loss as
\begin{align*}
    L_S(h, w) \texttt{=} -\Pr(z' = 1) \E{z' = 1}{\log h(x, m)} - \Pr(z' = 0) \E{x, m \in \Psim^{SR}}{w(x,  m) \log (1 - h(x, m))},
\end{align*}

where $z' = 1$ for $x, m \sim \mathcal{P}$ and $0$ for $x, m \sim \mathcal{P}(\cdot | y = 0)$. 
Next, define $w(x, m) = \frac{q(x , m | z = 1)}{q(x, m | z = 0)}$ and let $\hat{w}$ be the weight ratio learned by our model. Let $\hat{h} = \argmin{h \in \mathcal{H}} \hat{L}_S(h, \hat{w})$, and let $h^\star = \argmin{h \in \mathcal{H}} L(h, w^\star)$. Intuitively, $h^\star$ corresponds to the true difference between $\Pdata^{SR}$ and $\Pdata^{SR}(\cdot | y = 0)$. We can first decompose the generalization error as
\begin{align}
    L_S(\hat{h}, \hat{w}) - L_S(h^\star, w) &= [L_S(\hat{h}, \hat{w}) - \hat{L}_S(\hat{h}, \hat{w})] + [\hat{L}_S(\hat{h}, \hat{w}) - \hat{L}_S(h^\star, \hat{w})] \\
    &+ [\hat{L}_S(h^\star, \hat{w}) - \hat{L}_S(h^\star, w)] + [\hat{L}_S(h^\star, w) - L_S(h^\star, w)].
\end{align}

We know that $\hat{L}_S(\hat{h}, \hat{w}) \le \hat{L}_S(h^\star, \hat{w})$, so
\begin{align*}
    L_S(\hat{h}, \hat{w}) &- L_S(h^\star, w) \le |L_S(\hat{h}, \hat{w}) - \hat{L}_S(\hat{h}, \hat{w})| + |\hat{L}_S(h^\star, w) - L_S(h^\star, w)| \\
    &+ \hat{L}_S(h^\star, \hat{w}) - \hat{L}_S(h^\star, w) \\
    &\le \sup_{h, w} |L_S(h, w) - \hat{L}_S(h, w)| + |\hat{L}_S(h^\star, w) - L_S(h^\star, w)| + \hat{L}_S(h^\star, \hat{w}) - \hat{L}_S(h^\star, w).
\end{align*}

We first bound $\sup_{h, w} |L_S(h, w) - \hat{L}_S(h, w)|$. For notation, we rewrite $L_S(h, w)$ as $L_S(h, g)$, where $w(x,  m) = \frac{g(x, m)}{1 - g(x, m)}$ and $g$ belongs to some function class $\mathcal{G}$. Then, using Theorem 3.3 from~\cite{mohri2018foundations}, we get that $\sup_{h, w} |L_S(h, w) - \hat{L}_S(h, w)| \le 2 \mathfrak{R}_{n^{SR}}(\ell_S \circ \{\mathcal{H}, \mathcal{G}\}) + \sqrt{\frac{\log 1 / \delta}{2n^{SR}}}$ with probability at least $1 - \delta$, where $\ell_S \circ \{ \mathcal{H}, \mathcal{G}\}$ is defined as satisfying $\ell_S(h(x, m), g(x, m), y) = - y \log h(x, m) - (1 - y) \frac{g(x, m)}{1 - g(x, m)} \log (1 - h(x, m))$ for $h \in \mathcal{H}, g \in \mathcal{G}$.

Next, we bound $|\hat{L}_S(h^\star, w) - L_S(h^\star, w)|$. Let $W = \max w(x, m) < \infty$ be the maximum density ratio, and let $B_1 = \max_{x, m} \{-\log h^\star(x, m), -\log (1 - h^\star(x, m))\}$. Assume that $B_1 < \infty$.
We can apply standard concentration inequalities here (Hoeffding) to get that $|\hat{L}_S(h^\star, w) - L_S(h^\star, w)| \le W B_1 \sqrt{\frac{\log 2 /\delta }{2n^{SR}}}$ with probability at least $1 - \delta$. 

Finally, we bound $\hat{L}_S(h^\star, \hat{w}) - \hat{L}_S(h^\star, w)$. We can write $\hat{L}_S(h^\star, \hat{w}) - \hat{L}_S(h^\star, w)$ as
\begin{align}
    \hat{L}_S(h^\star, \hat{w}) - \hat{L}_S(h^\star, w) = \frac{1}{n^{SR}} \sum_{x \in \dsim_{SR}} (\hat{w}(x, m) - w(x, m)) \cdot (-\log (1 - h^\star(x, m))).
\end{align}

Define $\eta = \max (-\log (1 - h^\star(x, m))) \ge 0$ for $x, m \in \dsim_{SR}$, which is small as long as $h^\star(x, m)$ sufficiently classifies $x$ and is hence a property of how separated the reweighted simulation and true data is. Then,
\begin{align}
    |\hat{L}_S(h^\star, \hat{w}) - \hat{L}_S(h^\star, w)| \le \frac{\eta}{n^{SR}} \sum_{x, m \in \dsim_{SR}} |\hat{w}(x, m) - w(x, m) |.
\end{align}

Recall that $\hat{w}(x, m) = \frac{\hat{g}(x, m)}{1 - \hat{g}(x, m)}$ and $w(x, m) = \frac{g^\star(x, m)}{1 - g^\star(x, m)}$ where $g^\star(x, m) = \Pr(z = 1 | x, m)$, so $|\hat{w}(x, m) - w(x, m)| = \frac{|\hat{g}(x, m) - g^\star(x,  m)|}{(1 - \hat{g}(x, m))(1 - g^\star(x, m))}$. This denominator is greater than $(1 - \hat{g}_{\max})(1 - g^\star_{\max})$. Then,
\begin{align}
    |\hat{L}_S(h^\star, \hat{w}) - \hat{L}_S(h^\star, w)| \le \frac{\eta}{(1 - \hat{g}_{\max})(1 - g^\star_{\max}) n^{SR}} \sum_{x, m \in \dsim_{SR}} |\hat{g}(x, m) - g^\star(x, m) |.
\end{align}

We now look at the classifier for training $g$. The per-point cross entropy loss for $(x, m, z)$ is $\ell(g(x, m), z) = -\log g(x, m)$ for $z = 1$ and $-\log (1 - g(x, m))$ for $z = 0$. WLOG, assume for some $x$ and $m$, $g^\star(x, m) > \hat{g}(x, m)$. Then $|\ell(g^\star(x, m), 1) - \ell(\hat{g}(x, m), 1)| = \log \frac{g^\star(x, m)}{\hat{g}(x, m)} = \log \Big(1 + \Big(\frac{g^\star(x, m)}{\hat{g}(x, m)} - 1\Big)\Big) \ge \frac{g^\star(x, m) / \hat{g}(x, m) - 1}{g^\star(x, m) / \hat{g}(x, m)} = \frac{g^\star(x, m) - \hat{g}(x, m)}{g^\star(x, m)} \ge |g^\star(x, m) - \hat{g}(x, m)| $ and $|\ell(g^\star(x, m), 0) - \ell(\hat{g}(x, m), 0)| = \log \frac{1 - \hat{g}(x, m)}{1 - g^\star(x, m)} = \log \Big(1 + \Big(\frac{1 - \hat{g}(x, m)}{1 - g^\star(x, m)} - 1\Big)\Big) \ge \frac{(1 - \hat{g}(x, m)) / (1 - g^\star(x, m)) - 1}{(1 - \hat{g}(x, m)) / (1 - g^\star(x, m))} = \frac{g^\star(x, m) - \hat{g}(x, m)}{1 - \hat{g}(x, m)} \ge |g^\star(x, m) - \hat{g}(x, m)|$, where we use the inequality $\log(1 + x) \ge \frac{x}{1 + x}$ for $x > -1$. Therefore, with probability $1 - \delta$,
\begin{align*}
     |\hat{L}_S(h^\star, \hat{w}) &- \hat{L}_S(h^\star, w)| \le \frac{\eta}{(1 - \hat{g}_{\max})(1 - g^\star_{\max}) n^{SR}} \sum_{\mathclap{x, m \in SR}} |\ell(\hat{g}(x, m), z) - \ell(g^\star(x, m), z) | \\
     &\le \frac{\eta \Nsim^{SR}}{(1 - \hat{g}_{\max})(1 - g^\star_{\max}) n^{SR}} \bigg(\E{}{|\ell(\hat{g}(x, m), z) - \ell(g^\star(x, m), z) |} +  B_2 \sqrt{\frac{\log 2 / \delta}{2\Nsim^{SR}}}\bigg),
\end{align*}

where $B_2 = \max_{x, y} \{\ell (\hat{g}(x, m), z), \ell (g^\star(x, m), z) \} = -\log (\min \{\hat{g}_{\min}, g^\star_{\min}\})$. We assume that $B_2$ is finite, so there exists a constant $c$ such that
\begin{align*}
    |\hat{L}_S(h^\star, \hat{w}) - \hat{L}_S(h^\star, w)| &\le \frac{\eta \Nsim^{SR}}{(1 - \hat{g}_{\max})(1 - g^\star_{\max}) n^{SR}} \bigg(c |L(\hat{g}) - L(g^\star)| +  B_2 \sqrt{\frac{\log 2 / \delta}{2\Nsim^{SR}}}\bigg),
\end{align*}

where $L(g) = \E{x, m \in SR}{\ell(g(x, m), z)}$. Since $g^\star(x, m)$ is Bayes optimal, $|L(\hat{g}) - L(g^\star)| = L(\hat{g}) - L(g^\star) = L(\hat{g}) - \hat{L}(\hat{g}) + \hat{L}(\hat{g}) - \hat{L}(g^\star) + \hat{L}(g^\star) - L(g^\star) \le 2\sup_{g \in \mathcal{G}} |L(g) - \hat{L}(g)|$. From Theorem 3.3 in~\cite{mohri2018foundations}, this is bounded by $2 \mathfrak{R}_{n^{SB}}(\ell \circ \mathcal{G}) + \sqrt{\frac{\log 1/\delta}{2n^{SB}}}$ with probability at least $1 - \delta$. Then, applying a union bound, with probability $1 - \delta$, we have
\begin{align*}
    |\hat{L}_S(h^\star, \hat{w}) &- \hat{L}_S(h^\star, w)| \\
    &\le \frac{\eta \Nsim^{SR}}{(1 - \hat{g}_{\max})(1 - g^\star_{\max}) n^{SR}} \bigg(4c \mathfrak{R}_{n^{SB}}(\ell \circ \mathcal{G}) + 2c\sqrt{\frac{\log 2 /\delta}{2 n^{SB}}} +  B_2\sqrt{\frac{\log 4 / \delta}{2\Nsim^{SR}}}\bigg).
\end{align*}

Putting everything together with another union bound, with probability $1 - \delta$, the generalization error is at most
\begin{align}
    L_S(\hat{h}, \hat{w}) &- L_S(h^\star, w) \le 2\mathfrak{R}_{n^{SR}}(\ell_S \circ \{ \mathcal{H}, \mathcal{G}\}) + (1 + WB_1) \sqrt{\frac{\log 8 / \delta}{2n^{SR}}} \\
    &+ \frac{\eta \Nsim^{SR}}{(1 - \hat{g}_{\max})(1 - g^\star_{\max}) n^{SR}} \bigg(4c \mathfrak{R}_{n^{SB}}(\ell \circ \mathcal{G}) + 2c  \sqrt{\frac{\log 4 /\delta }{2 n^{SB}}}+ B_2 \sqrt{\frac{\log 8 / \delta}{2\Nsim^{SR}}} \bigg).
\end{align}

\end{proof}

\section{Experiment Details}\label{supp:exp}

\subsection{\multicwola Experiments}

For the \multicwola experiment, we used the anomaly and simulation data from the Pythia 8 simulations in the LHC Olympics Dataset to create an unlabeled dataset we want to perform anomaly detection on~\cite{Kasieczka:2021xcg}. We have $k = 3$, and construct $M_i(m)$ based on the thresholds $[[3.3, 3.7], [0.09, 0.13], [0.3, 0.35]]$ on the first three features.
 For standard \cwola, only the third feature is regarded as the resonant feature, and it is thresholded with the interval $[0.3, 0.35]$. We constructed training datasets of varying sizes with class balance $\Pr(y=1) = 0.149$. We used one test dataset with $65755$ randomly sampled anomaly points and $161658$ randomly sampled background points.
 
All methods were trained using scikit-learn's MLPClassifier with max\_iter=5000. For $\multicwola$'s weak supervision step, we learn the parameters of the graphical model using SGD and PyTorch~\cite{paszke2019pytorch} with class balance $\Pr(y = 1) = 0.25$, $30000$ epochs, and learning rate $=1e-6$.

\subsection{\multisalad Experiments}

\paragraph{Setup} 
We use MLPs from Keras~\citep{chollet2015keras}, each with $3$ hidden layers of dimension $32$, ReLu activation, and trained with cross-entropy loss and the Adam optimizer. We train for $50$ epochs, batch size $200$, and default parameters otherwise. Finally, we evaluate our approach on a new test set containing $200000$ background points and $200000$ anomaly points. This test set is used to produce the signal efficiency to rejection rate.
All experiments were run on a personal laptop.

\paragraph{Additional Results}

In Figure~\ref{fig:runs}, we show our results on individual runs. This is because computing the confidence intervals of these curves averaged across the $10$ random runs is too noisy due to the magnitude of the reciprocal 1/FPR.

\begin{figure}[t]
    \centering
    \includegraphics[width=0.44\textwidth]{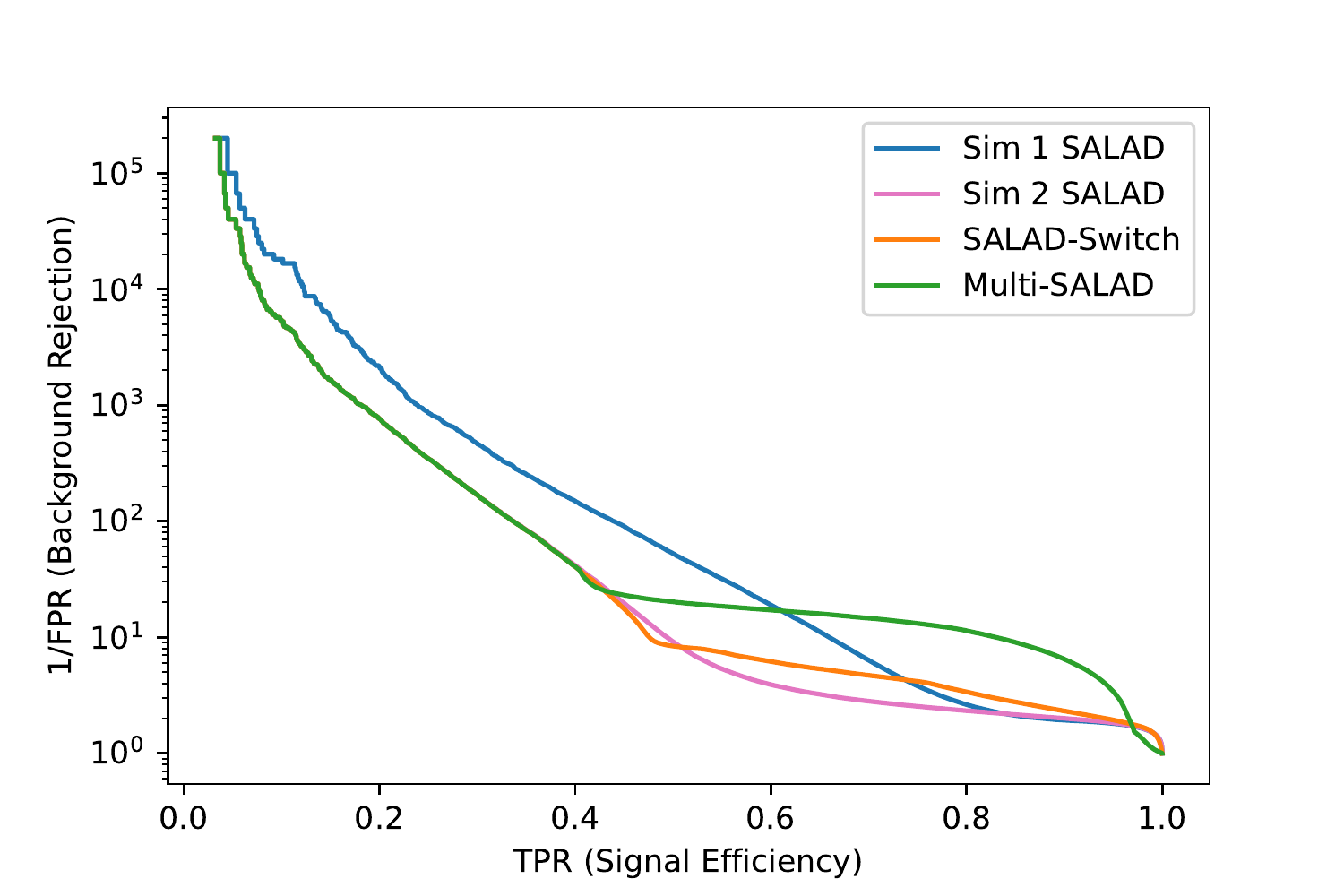}
    \includegraphics[width=0.44\textwidth]{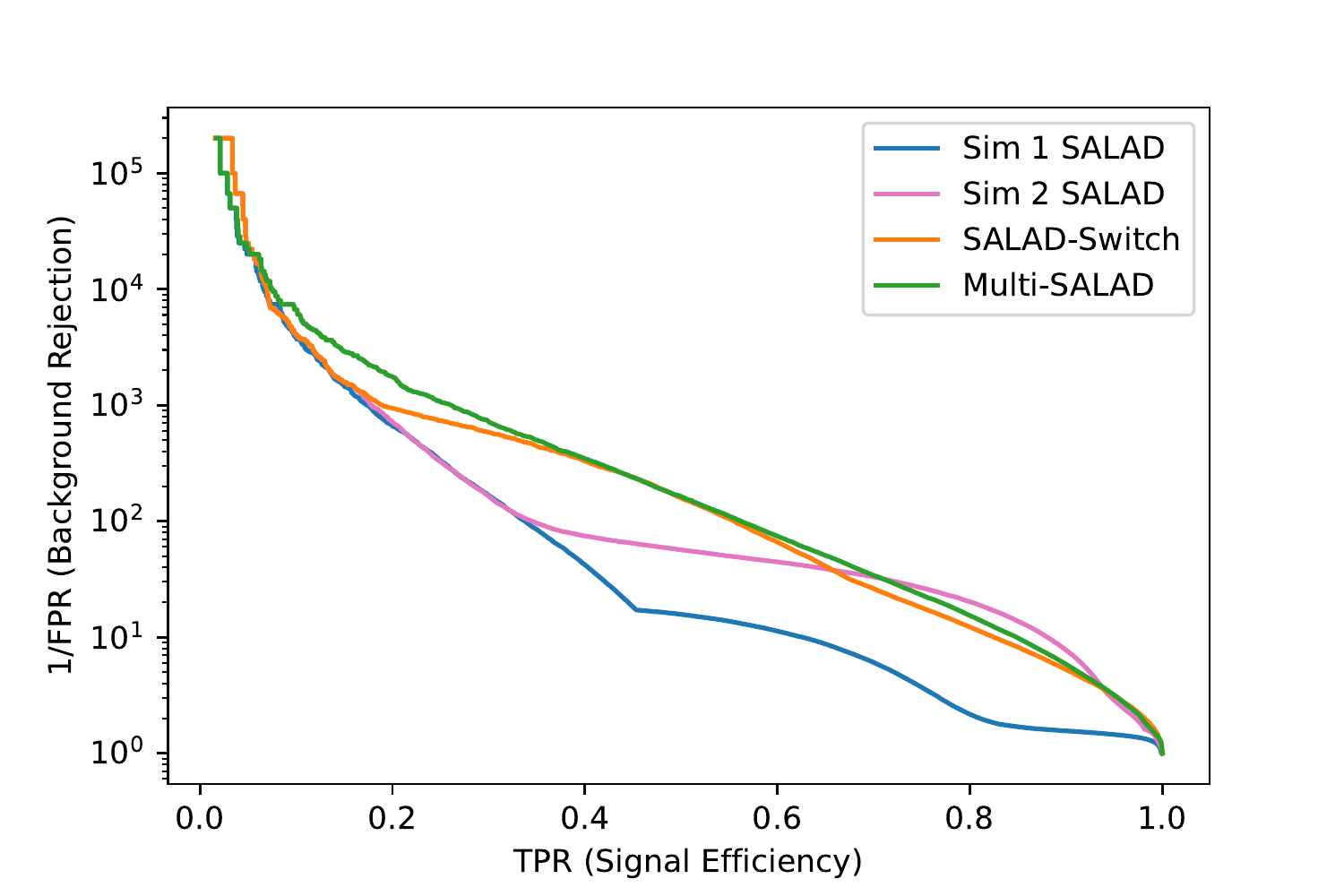}
    \includegraphics[width=0.44\textwidth]{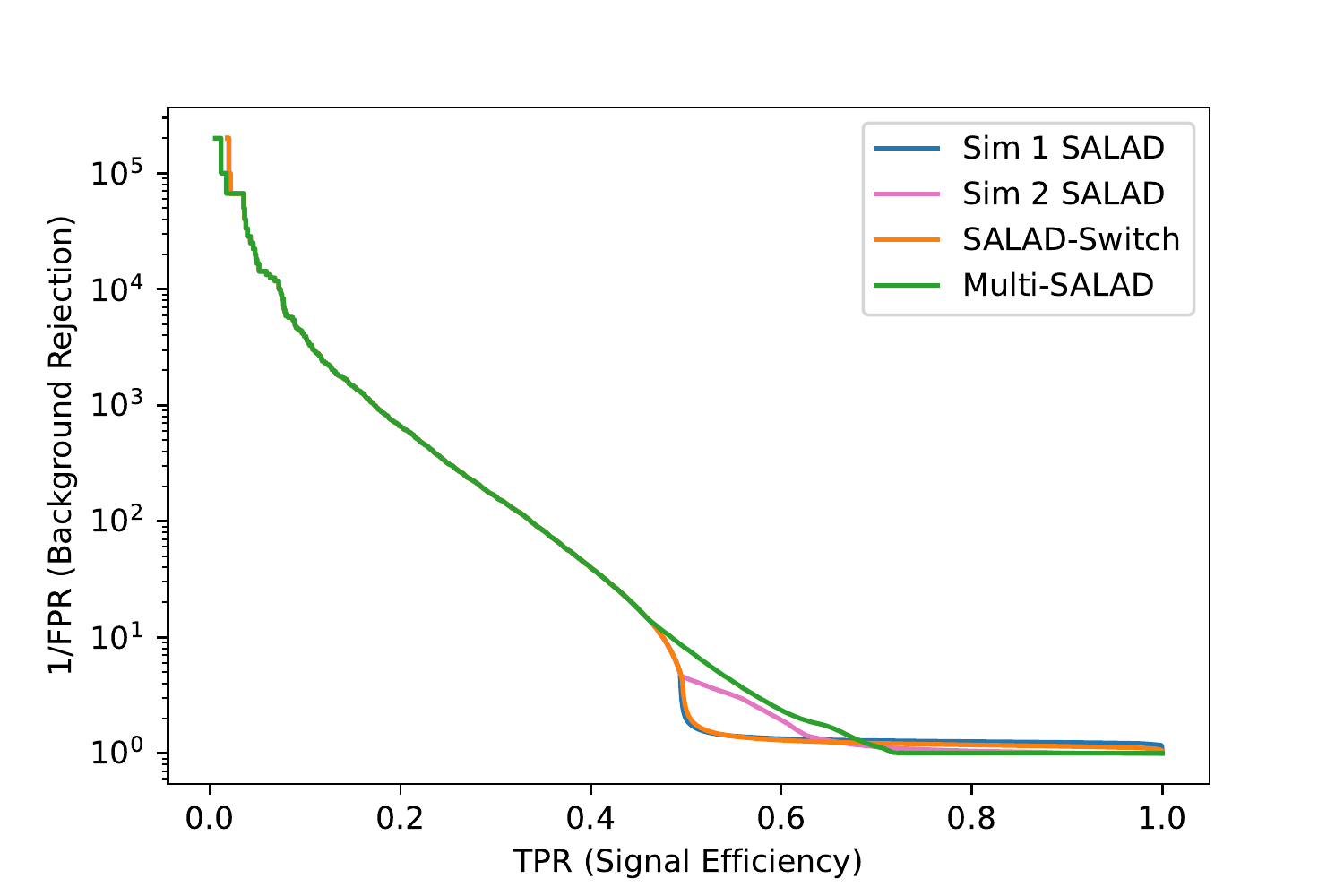}
    \includegraphics[width=0.44\textwidth]{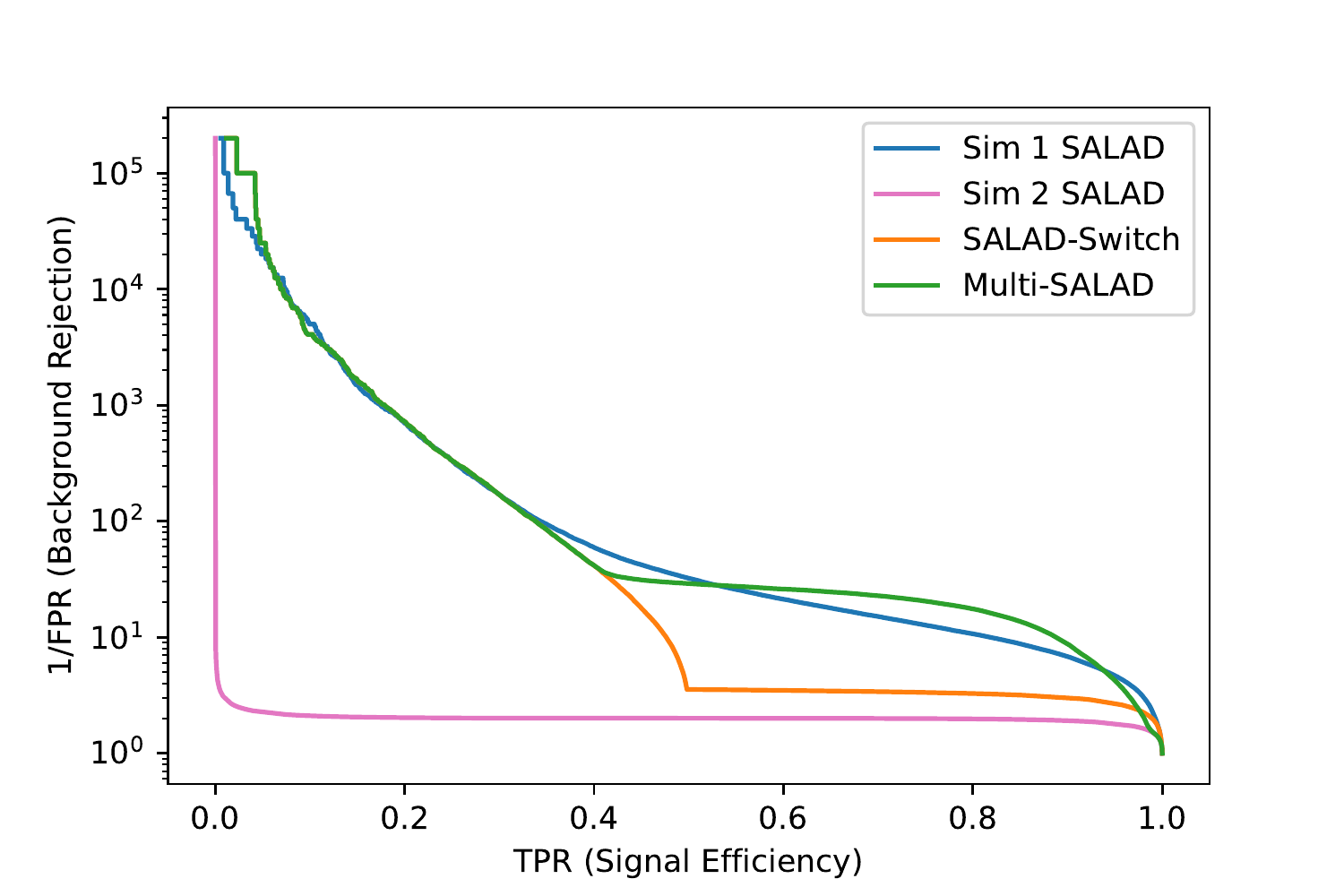}
    \includegraphics[width=0.44\textwidth]{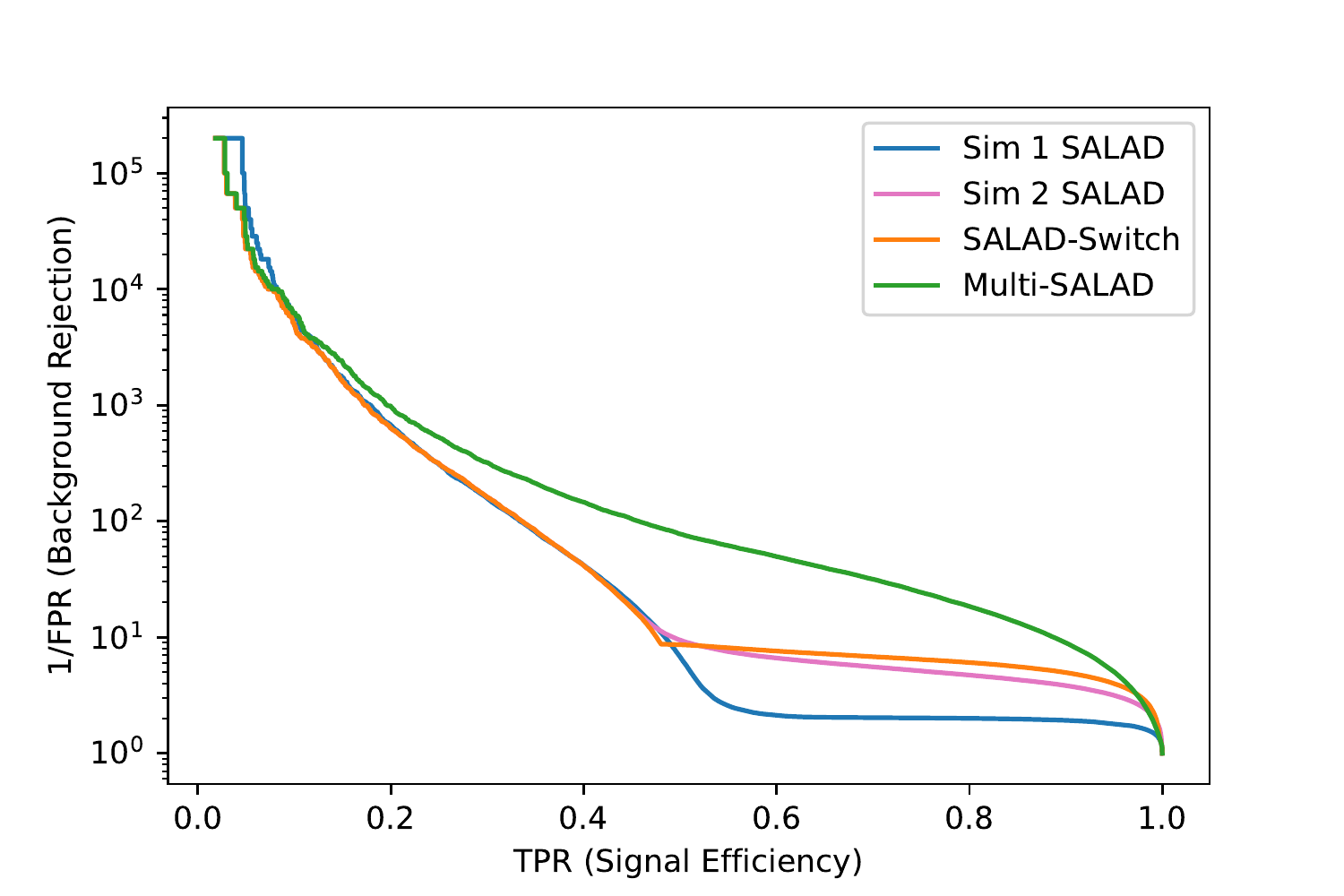}
    \includegraphics[width=0.44\textwidth]{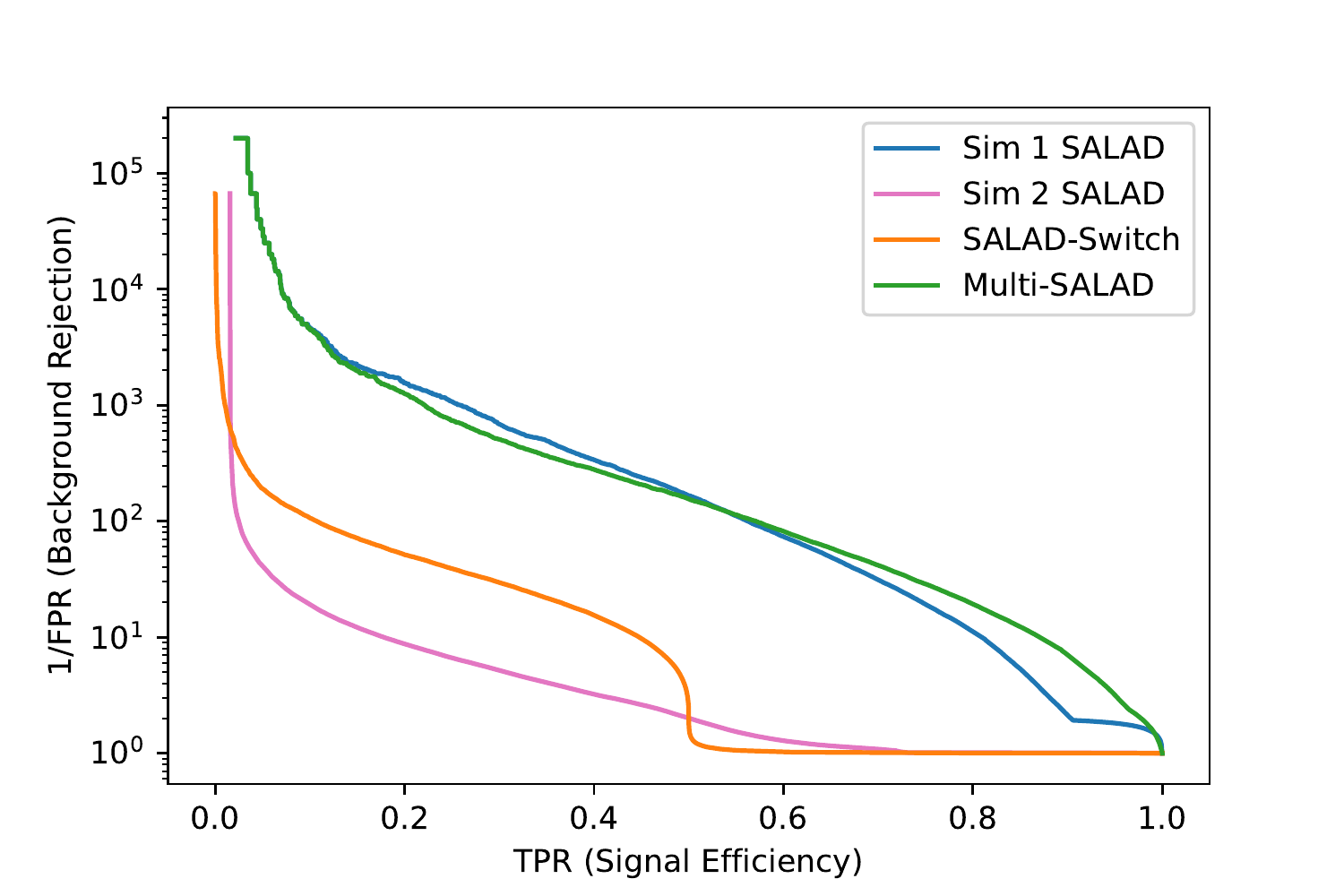}
    \includegraphics[width=0.44\textwidth]{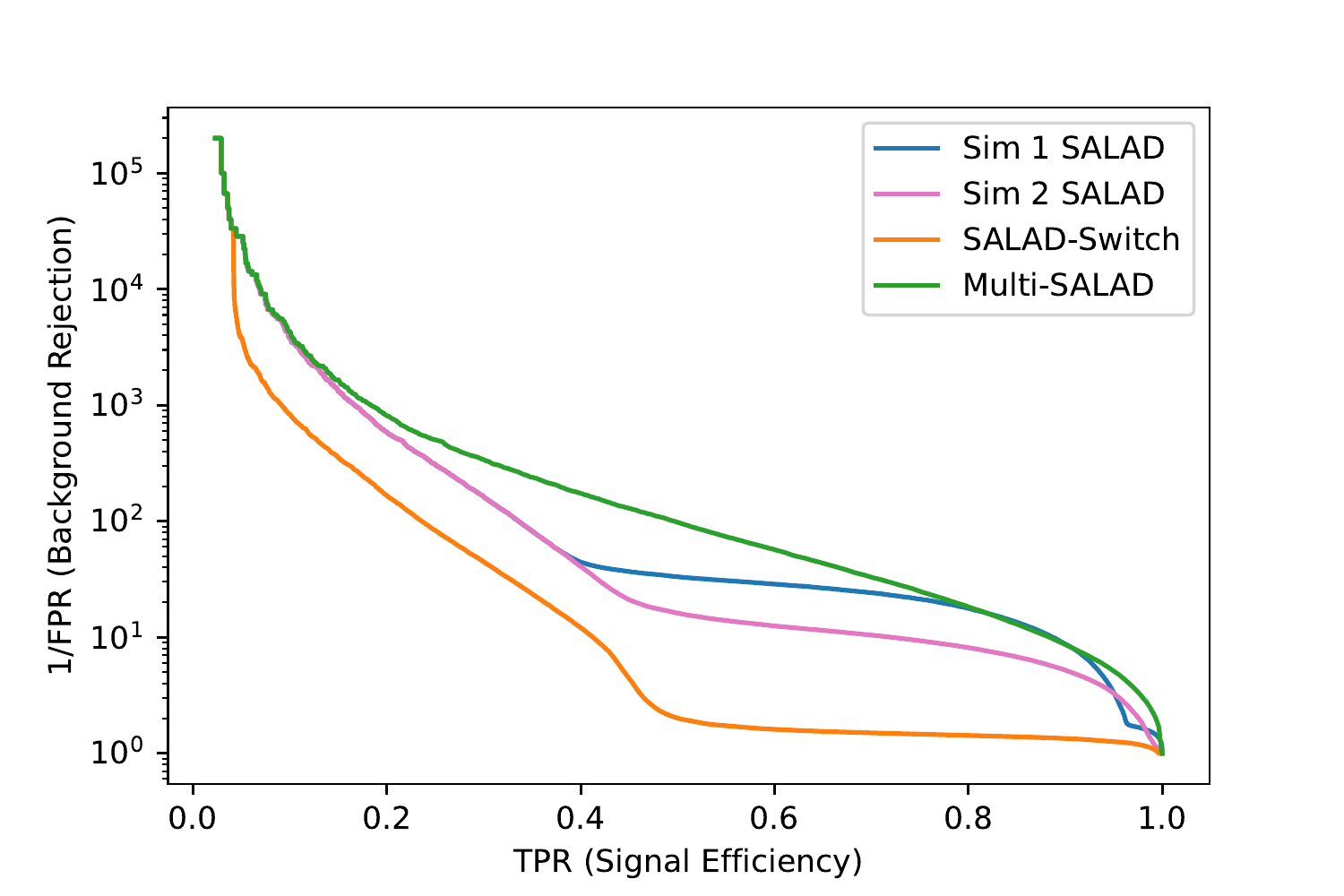}
    \includegraphics[width=0.44\textwidth]{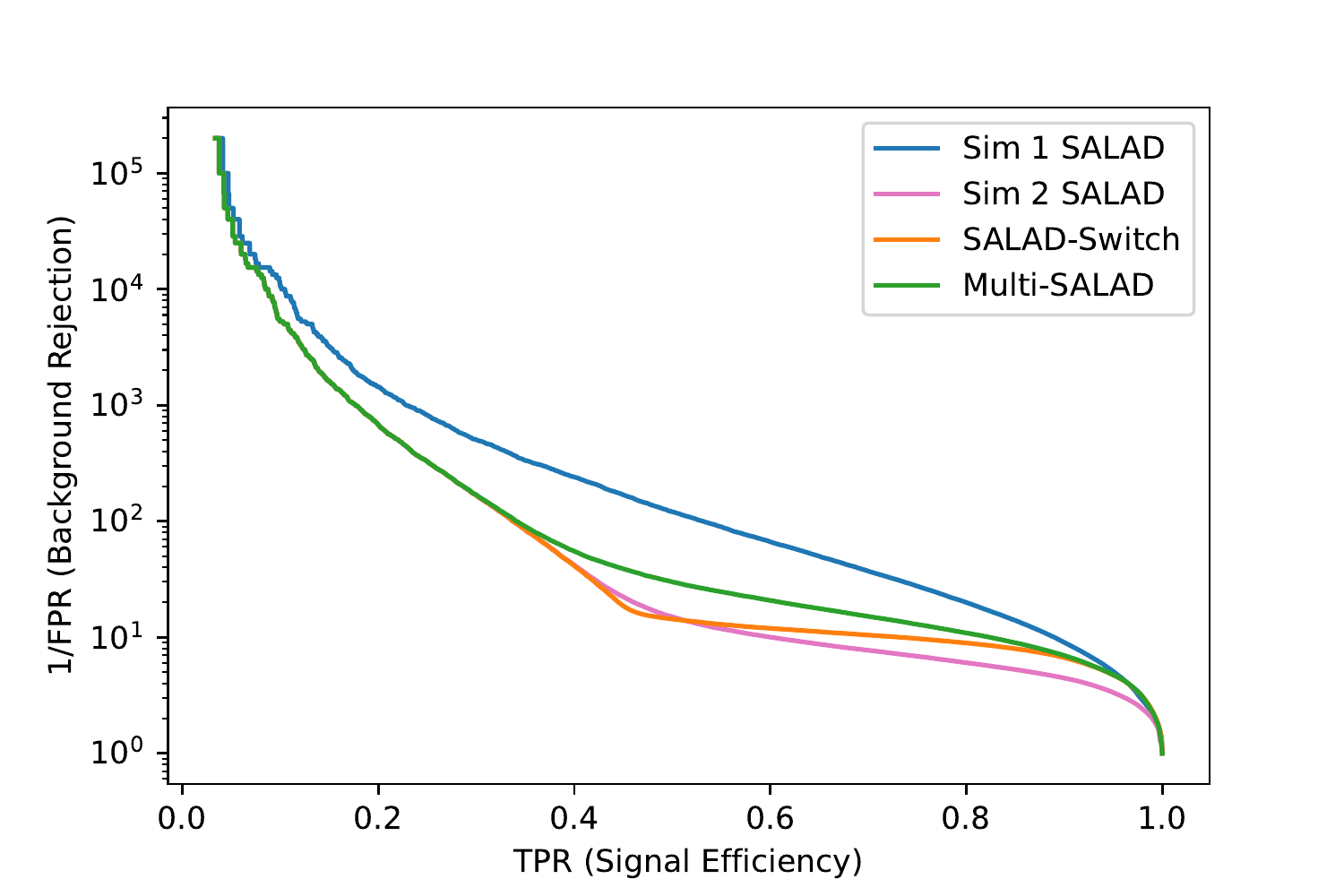}
    \includegraphics[width=0.44\textwidth]{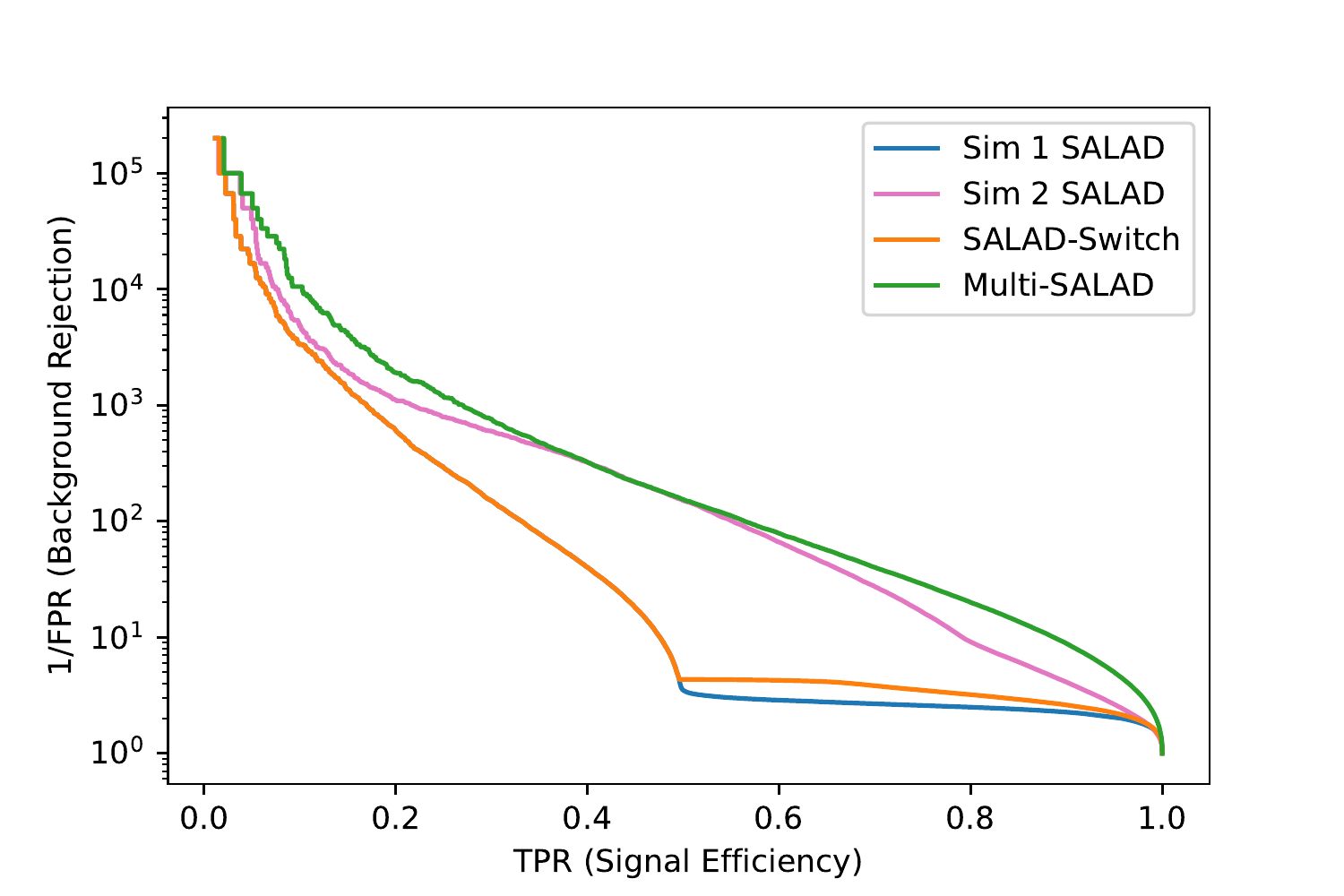}
    \includegraphics[width=0.44\textwidth]{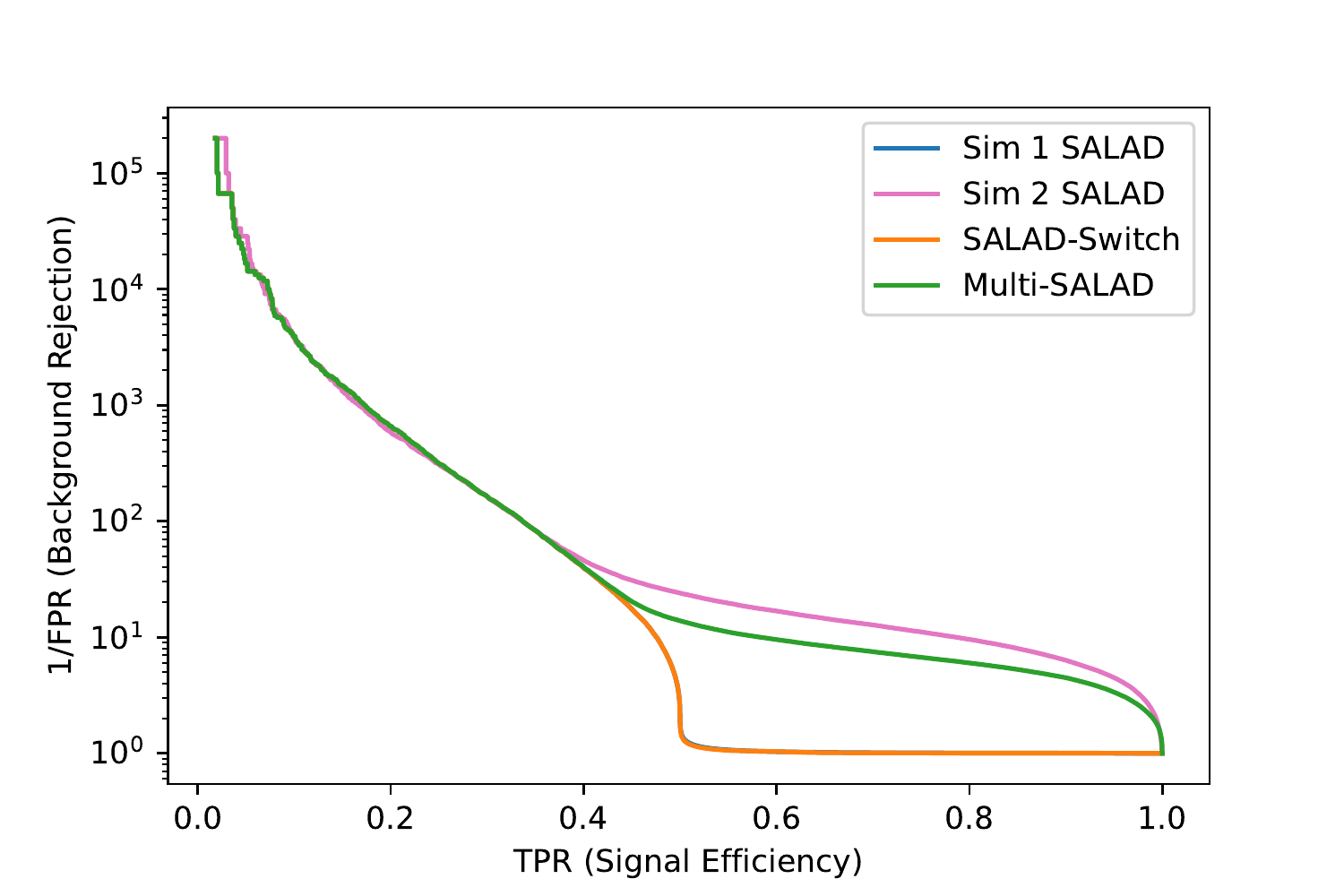}

    \caption{Results on individual runs.}
    \label{fig:runs}
\end{figure}

\end{document}